\documentclass[runningheads,a4paper]{llncs}

\usepackage{amsmath}
\usepackage{amssymb}
\usepackage{theorem}
\usepackage{theorem}
\usepackage{epsfig}
\usepackage{graphicx}
\usepackage{epstopdf}
\usepackage{t1enc}
\usepackage[utf8]{inputenc}
\usepackage[T1]{fontenc}
\usepackage{hyperref}
\newcommand{\1}{\hspace{0 mm}^\sim \hspace{-0.2 mm}}
\newcommand{\N}{\mathbb{N}}
\newcommand{\Z}{\mathbb{Z}}

\def\fagnotbb{{\triangleright}}
\def\fagnotff{{\triangleleft}}
\spnewtheorem*{sketch}{Sketch of the proof}{\itshape}{\rmfamily}
\newcommand\fagnotins[2]{{}^{#1\!}B^{#2}}
\newcommand\fagnotwrd[2]{{}_{#1}W_{#2}}
\renewcommand\fagnotwrd[2]{{}_{#1\!}W_{\! #2}}
\newcommand\fagnotrul[4]{\langle  #1 {\mid} #3{-}#4  {\mid} #2  \rangle}
\newcommand{\fagnotproduc}[4]{#2,#3\, \vdash_{#1} #4} 
\newcommand{\fagnotlang}[1]{L({#1})} 
\newcommand{\fagnotcir}[1]{{}^\sim #1}

\DeclareMathOperator{\Card}{Card}

\title{Splicing Systems from Past to Future: \\ Old and New Challenges
\thanks{
Partially supported by
the $FARB$ Project
{\it ``Aspetti algebrici e computazionali nella teoria dei codici
e dei linguaggi formali''}
(University of Salerno, 2012), the
$FARB$ Project
{\it ``Aspetti algebrici e computazionali nella teoria dei codici,
degli automi
e dei linguaggi formali''}
(University of Salerno, 2013)
and the $MIUR$ $PRIN$ $2010$-$2011$ grant
{\it ``Automata and Formal Languages: Mathematical and Applicative Aspects''},
code H41J12000190001.}}
\author{Luc Boasson\inst{1} \and Paola Bonizzoni\inst{2} \and Clelia De Felice\inst{3} \and
Isabelle Fagnot\inst{4} \and \\ Gabriele Fici\inst{5} \and Rocco Zaccagnino\inst{3} \and
Rosalba Zizza\inst{3}}
\authorrunning{L.~Boasson et al.}
\institute
{Universit\'e Paris Diderot--Paris VII, LIAFA,  \and Universit\`a degli Studi di Milano Bicocca, \\
Dipartimento di Informatica, Sistemistica e Comunicazione,
\and Universit\`a degli Studi di Salerno, Dipartimento di Informatica,
\and Universit\'e Paris Est, LIGM and Université  Paris Diderot--Paris VII,
\and Universit\`a degli Studi di Palermo,
Dipartimento di Matematica e Informatica}

\begin{document}

\sloppy

\maketitle

\begin{abstract}
A splicing system is a formal model of a recombinant behaviour of sets of
double stranded DNA molecules when acted on by restriction enzymes and ligase.
In this survey we will concentrate on a specific behaviour of a type of splicing systems,
introduced by P\u{a}un
and subsequently developed by many researchers in both linear and circular case
of splicing definition. In particular, we will present recent results on this topic
and how they stimulate new challenging investigations.
\end{abstract}

\section{Introduction} \label{introduction}

Linear splicing  is a language-theoretic word operation introduced by T. Head in  \cite{h87}
which models a DNA recombination process, namely  the action
of two compatible restriction enzymes and a ligase enzyme on two DNA
strands. The first enzyme recognizes a specific pattern in any DNA string and
cuts the string containing this pattern in a specific point inside the
recognized pattern. The second restriction enzyme acts similarly.
The site and the shape of the cut ends are
specific to each enzyme. Then ligase enzymes perform the second step of
this process, pasting together properly matched fragments,
under some chemical conditions. Abstracting this phenomenon, the
linear splicing operation is applied to two words and two different words may be generated.
It concatenates a prefix of one string
with a suffix of another string, under some conditions, represented
as a (splicing) rule. In this paper a rule will be represented as
a quadruple of words $r=u_1\#u_2\$u_3\#u_4$. A linear splicing system consists of a set $I$
of words (called initial language) and a set $R$ of (splicing) rules.
The language generated by a splicing system (splicing language)
contains every word that can
be obtained by repeated application of rules to pair of words in the initial language
and to the intermediately produced words
(each word in $I$
or obtained by splicing is present in an
unbounded number of copies).
Many of the initial research has been devoted to
the study of the computational power of linear splicing
systems. It mainly depends on
which level of the Chomsky hierarchy
$I$ and $R$ belong to. Moreover, for any pair $F_1$, $F_2$ of families
of languages in the Chomsky hierarchy, the class of languages generated by
linear splicing systems with $I$ in $F_1$ and $R$ in $F_2$ is
either a specific class of languages
in the Chomsky hierarchy
or it is strictly intermediate
between two of them \cite{hb,book}.
For instance, the class of languages generated by splicing systems
with a finite initial language and a finite set of rules,
often referred to as finite splicing systems, contains all finite languages and it is
strictly contained in the class of regular languages.
This basic result has been proved
in several  papers by using different approaches (see \cite{hb,z}).
In this paper by a linear splicing system we always mean a finite linear splicing system.
In spite of the vast literature on the topic, many problems are still open.
For instance, the question of whether or not one of the known subclasses of the regular languages
corresponds to the class of splicing languages
is still unanswered. In \cite{bj} a property of these
languages is provided: a necessary condition for a regular
language to be a splicing language is that it must have a constant. However a
structural characterization of splicing languages is still lacking. It has been
provided for splicing languages generated by reflexive and/or symmetric rules
\cite{b,rifl}.
Moreover, it has been proved that it
is decidable whether a regular language is generated by a splicing system \cite{kk}.

Circular splicing systems were introduced in \cite{h92}
along with various open
problems related to their computational power.
In the circular context, the splicing operation acts on two circular DNA molecules
by means of a pair of restriction enzymes as follows.
Each of these two enzymes is able to recognize a pattern inside one of the given
circular DNA molecules and to cut the molecule in the middle of such a pattern.
Two linear molecules are produced and
then they are pasted together by the action of ligase enzymes.
Thus a new circular DNA sequence is generated. For instance, circular splicing models
the integration of a plasmid into the DNA of a host bacterium \cite{hsugg,h06,h12,hplasmid}.
Depending on whether or not these ligase enzymes substitute the recognized pattern
(in nature, both situations can happen), we have the Pixton definition or the Head
and  P\u{a}un definition
(see Section \ref{cdef}).

Obviously a string of circular DNA can be represented by a circular word,
i.e., by an equivalence class with respect to
the conjugacy relation $\sim$, defined by $xy \sim yx$, for $x,y \in A^*$
\cite{lot}.
The circular splicing operation is applied to two circular words and
a circular word may be generated.
A circular language is a set of
circular words and we can also
give a definition of a regular (resp. context-free, context-sensitive)
circular language.
In short, as in the linear case, circular splicing rules are iteratively applied starting
from an initial circular language and circular words in this
process are supposed to be present in an
unbounded number of copies.
A circular splicing system is defined by a circular language $I$, the initial
set, and a set of splicing rules $R$. Once again, rules will be represented as
quadruple of words
$r=u_1\#u_2\$u_3\#u_4$.
The circular language generated
by a circular splicing system is the smallest language which contains
$I$ and is invariant
under iterated splicing
by rules in $R$.

While there have been many
articles on linear splicing,
relatively few works on circular splicing
systems have been
published. Once again, the
computational power of circular splicing systems
depends on the level in
the (circular) Chomsky hierarchy the initial set $I$
and
the set $R$ of rules belong to.
In this paper $R$ will always be a finite set.
Moreover we
mainly focus on finite circular splicing systems.
A circular splicing system is finite (resp. regular,
context-free, context-sensitive) if its initial set is finite
(resp. regular, context-free, context-sensitive).
The interesting fact is that
in contrast with the linear case,
the class of circular
languages generated by finite circular splicing systems contains
context-free circular languages
which are
not regular \cite{ssd}, context-sensitive
circular languages which
are not context-free \cite{flat},
and it does not contain all regular circular languages \cite{rairo}.
Recently, it has been proved that this class is contained
in the class of
context-sensitive circular languages
(see \cite{flat}).
This result remains true if we consider context-sensitive circular
splicing systems (see \cite{flat}).
In the same paper it has also been proved that
the language is context-free if it is generated by an alphabetic
context-free splicing system
(a circular splicing system is alphabetic if for any rule
$r=u_1\#u_2\$u_3\#u_4$ the words $u_j$ are letters or the empty word).
It is decidable whether a regular circular language and the language generated by
a regular circular splicing system are equal but we do
not still known how to decide the inclusion of one of them in the other \cite{flat}.
It is also decidable whether a given regular language is generated by an alphabetic
finite splicing system
\cite{flat} but the same problem is still open for general systems.
All the above mentioned results from \cite{flat} have been obtained first for a new variant of circular splicing,
introduced in the same paper and named flat splicing, then easily extended to the classical model.
Other still unanswered questions remain. For instance we still do not know how to decide
whether a circular splicing language is regular, nor do we know how to decide
whether a regular circular language is a splicing language.
Finally, a characterization of those splicing languages which are
regular, or vice versa is still lacking.
Partial contributions to these questions have been given in \cite{completi,rairo,marked}.
In particular, regular circular languages generated by monotone complete systems
have been characterized in \cite{completi}. This result will be differently proved in Section
\ref{Scomplete}, by means of some special flat systems.

The paper is organized as follows.
In Sections \ref{wdef}, \ref{cwdef}, we set up
the basic definitions for words and circular words respectively.
We briefly discuss linear splicing in Sections \ref{ldef}--\ref{mainlinear}.
In Sections \ref{ldef} and \ref{lsplicing}, we introduce preliminary concepts
whereas a simple type of splicing is discussed in Section \ref{sh}.
The main problems and recent results are presented in Section \ref{mainlinear}.
Sections \ref{cdef}--\ref{Scomplete} are devoted to circular splicing.
In Section \ref{cdef}, we set up
the basic definitions on circular splicing.
Decidability properties and results on the position of splicing systems and their languages
within the Chomsky hierarchy are presented in Section \ref{Sflat}, along with
flat splicing.  The special case of the alphabetic splicing systems will
be discussed in Section \ref{Sflat}. A subclass of them, namely
the class of the so-called simple or semi-simple systems,
have been considered in the literature and their definition
is given in Section \ref{cssh}. A known characterization of regular languages generated by
special semi-simple systems (marked systems, complete systems) is
recalled in Sections \ref{Smarked}, \ref{Scomplete}.

\section{Words} \label{wdef}

We suppose the reader familiar with basic notions on formal languages
and we provide here only the necessary notations, \cite{bpr,hum,lot,lotN}
are some general references on the topic.
Let us denote $A^*$ the free monoid over a finite alphabet $A$
and $A^+=A^* \setminus 1$, where $1$ is the empty
word. For a word $w \in A^*$, $|w|$ is the length of $w$ and
for every $a \in A$, $w \in A^*$, we denote by $|w|_a$ the
number of occurrences of $a$ in $w$.
We also set $alph(w) = \{a \in A ~ | ~ |w|_a > 0 \}$
and $|w|_{A'} = \sum_{a \in A'} |w|_a$, where
$A' \subseteq A$.
A word $x \in A^*$ is a {\it factor} of $w \in A^*$ if
there are $u_1,u_2 \in A^*$ such that $w=u_1xu_2$.
We denote by $Fact(w)$ the set of the factors of $w$.
The {\it reversal} $w^{Rev}$
of $w$ is defined by the relations $1^{Rev} = 1$ and,
for all $x \in A^*$, $a \in A$, $(xa)^{Rev} = ax^{Rev}$.
For a subset $X$ of $A^*$, $X^{Rev} = \{w^{Rev} ~|~ w \in X \}$
is the reversal of $X$.
Furthermore, for a subset $X$ of $A^*$,
$\Card(X)$ is the cardinality of $X$.
A language is {\it regular} if it is recognized by a finite automaton.
We denote by $Fin$ (resp. $Reg$) the class of finite
(resp. regular) languages over $A$, at times represented
by means of regular expressions.
Let $L \subseteq A^*$ be a set and let $w \in A^*$ be a word.
We denote by $\Gamma_L(w)$ the set of {\em contexts} of $w$ in $L$, that is
$\Gamma_L(w)=\{(u,v) \in A^* \times A^* ~|~ uwv \in L \}$.
A word $w \in A^*$ is said to be a {\it constant} for $L$ if
for any $(u,v), (u',v') \in \Gamma_L(w)$ one has also
$(u,v'), (u',v) \in \Gamma_L(w)$ \cite{bpr,sch}.

\section{Splicing Operation from Scratch} \label{ldef}

As a model of the  biochemical operation of splicing, in \cite{h87}
Head considered the following
string operation (passing from double stranded sequences to strings is allowed, due
to the precise Watson-Crick complementarity of nucleotides). Consider an alphabet $A$
and two finite sets, $B$ and $C$, of triples
$(\alpha,\mu,\beta)$ of words in $A^*$.
These triples are called patterns and the string
$\mu$ is called the crossing of the triple. Given two patterns
(with the same crossing) $(\alpha, \mu ,\beta)$,
$(\alpha' , \mu,  \beta')$,
both in $B$ or both in $C$,
and two words
$u \alpha \mu \beta v$, $p \alpha' \mu \beta' q$, these words can be
spliced and the splicing operation produces $u\alpha \mu \beta' q$,
$p \alpha' \mu \beta v$.

\begin{example}
{\rm Let us consider the word
$w=cxcxc$ and the triple $(c,x,c)$. The pattern $cxc$ occurs twice in $w$
and the triple can be applied, coupled with itself, to two copies of $w$.
Thus, using all combinations, the result of splicing is
$(cx)c+(cx)^2c+(cx)^3c$.}
\end{example}

Abstracting further from this idea, P\u{a}un considered splicing rules of the
form $r=u_1\#u_2\$u_3\#u_4$,
where $u_1,u_2,u_3,u_4$ are strings over a given alphabet $A$ and $\#,\$ \not \in A$
\cite{Paun96}.
The words $u_1u_2$, $u_3u_4$ are called {\em sites} of $r$.
Given such a rule $r$, by splicing the two strings $x=x_1u_1u_2x_2$,
$y=y_1u_3u_4y_2$,
the strings $w'=x_1u_1u_4y_2$,
$w''=y_1u_3u_2x_2$ are produced.
We denote this operation by
$(x,y) \vdash_r (w',w'')$.

It is clear that this is a generalization of Head's definition
of splicing
(where the crossing is supposed to be empty).

\begin{example} \cite{hp06} \label{ExCG}
{\rm
Let $r=cg \# cg \$ cg \# cg$ and consider
$u=aacgcgaacgcgaa=(aacgcg)^2aa$ and $v=ttcgcgtt$.
There are two occurrences of the string $cgcg$ in $u$ and only one in $v$.
Thus, $aacgcgtt, aacgcgaacgcgtt$ are generated as well as
$ttcgcgaacgcgaa, ttcgcgaa$: the former by applying $r$ to $u$ and $v$,
the latest by applying $r$ to $v$ and $u$.}
\end{example}

A still more general definition of splicing was considered
by Pixton \cite{pixDAM}.
The rules are of the form $(\alpha, \alpha', \beta)$
and by splicing two strings $\epsilon \alpha \delta$
and $\epsilon' \alpha' \delta'$,
the strings $\epsilon \beta \delta'$,
$\epsilon' \beta \delta$ are generated.
Observe that this definition of splicing is more general than P\u{a}un's
one (note that the context substrings $\alpha, \alpha'$
are substituted by $\beta$ during the splicing).

\begin{example}
{\rm
The rule $(a,xa,xa)$ applied to (two copies of) the word
$cxae$ generates $cxxae$ (and $cxae$).
Splicing allows us to ``pump'' the letter $x$. }
\end{example}

\section{Computing Devices Based on Splicing: Splicing Systems} \label{lsplicing}

{\em Splicing systems} are models for generating languages based on the splicing operation.
In the literature, different models of splicing systems were presented
\cite{hb,Paun96,book,pixDAM}
and three kinds of splicing systems have been studied,
by using the three definitions of splicing operation
given in the previous section.
In this paper we consider only the {\em iterated splicing}
operation given by  P\u{a}un and the corresponding systems as follows.

A {\em splicing system} (or $H$-system)
is a triple $H=(A,I,R)$, where $A$ is a
finite alphabet, $I \subseteq A^*$
is the initial language and $R$ is the set of
rules, with
$R \subseteq A^* \# A^* \$ A^* \# A^*$
and $\#, \$ \not \in A$.
It is finite if $I$ and $R$ are both finite sets.
Let $L \subseteq A^*$.
We set
$\sigma'(L)=\{w',w'' \in A^* ~|~
(x,y){\vdash}_r ~(w',w''),
~x,y \in L, r \in R\}$.
The (iterated) splicing
operation is defined as follows
\begin{eqnarray*}
\sigma^0(L) &=&L,   \\
\sigma^{i+1}(L) &=& \sigma^i(L) \cup
\sigma'(\sigma^i(L)), ~i \geq 0, \\
\sigma^*(L) &=& \bigcup_{i \geq 0} \sigma^i(L).
\end{eqnarray*}

\begin{definition}[P\u{a}un splicing language]
Given a splicing system $H=(A,I,R)$,
the language
$L(H)=\sigma^*(I)$ is the language
generated by $H$.
A language $L$ is $H$
{\em generated} (or is a {\em  P\u{a}un splicing language})
if a splicing
system $H$ exists such that $L=L(H)$.
\end{definition}

We have adopted the more realistic operation of splicing defined
by taking into account both of the two possible words obtained by recombination
and starting with two words and a rule.
This operation is also known as 2-splicing.
A different definition can be obtained when we take into account only one word
(1-splicing). Relations between the computational power of splicing
systems with 2-splicing and splicing systems with 1-splicing can be found in
\cite{book,vz}.

In order to characterize regular languages generated by finite
splicing systems, some partial results have been provided by
considering (realistic) additional hypotheses or
suitable restrictions on splicing rules. As an example, the
reflexive hypothesis, or symmetry. We recall that
$R$ is
{\it reflexive} if for each
$u_1 \# u_2 \$ u_3 \# u_4$ in $R$, we have
$u_1 \# u_2 \$ u_1 \# u_2$ and $u_3\#u_4 \$ u_3 \#u_4 \in R$.
$R$ is {\it symmetric} if for each
$u_1 \# u_2 \$ u_3 \# u_4 \in R$, we have
$u_3 \# u_4 \$ u_1 \#u_2 \in R$.
Observe that 2-splicing is equivalent to 1-splicing plus the symmetric hypothesis on $R$.

\section{Simple Systems: the Origins and Stimulated (Related) Results} \label{sh}

The splicing operation was explicitly linked with the concept of
constant already by Head \cite{h87}, in his seminal paper.
Indeed, it is evident the similarity between a constant
and a crossing in Head's definition.
Head proved that {\em persistent} splicing languages
coincide with
Strictly Locally Testable (SLT) languages.
In addition, he proved that
SLT languages may be generated by systems such that $B=C=A^k$ , for $k \geq 1$
(uniform splicing systems or Null Context H-systems, NCH systems).
To do this, he used the result of
De Luca and Restivo (1980) showing that a language $L$ is SLT if and only if
there is an integer $k$ such that all strings in $A^k$
are constants with respect to $L$.

In \cite{h98slt}, Head gave different characterizations of the family of SLT languages,
pointing out that the class of SLT
languages itself is the union of the
families of languages generated by a special hierarchy of SH systems, splicing systems
which are a subclass of
NCH systems in \cite{mat}
(each crossing of a triple is a letter).
He gave a procedure which, for a regular language $L$,
determines whether $L$ is SLT.  When $L$ is SLT,
this procedure specifies constructively
the smallest family in the hierarchy containing $L$.

The restrictive class of {\em simple systems (SH)} $G$,
was explicitly introduced in \cite{mat}, based on rules
of the form $a\#1\$a\#1$, $a \in A$, i.e.,
splicing is allowed on every position
where such a symbol (marker) appears.
Clearly each language $L(G)$ is regular
and since they are special NCH systems, we have that $L(G)$ is SLT.
A characterization of languages in SH is also provided and,
in the case of unary languages, they have a very simple
regular expression ($L=a^*$ or $L=a^+$).
SH systems were subsequently studied in \cite{ceterchi},
also  by considering different positions of the letter $a$ inside a rule.

In 2001,
SH systems were generalized by considering {\em semi-simple splicing
systems SSH}, where all rules have the form
$a_i\#1\$a_j\#1$, $a_i,a_j \in A$ \cite{ssss}.
Also in this case, four types of rules can be considered,
depending on the position of the two letters.
In \cite{ssss} only (1,3)-SSH are considered, i.e.,
when all rules have the form $a_i\#1\$a_j\#1$, $a_i,a_j \in A$ ,
and the main result
is a characterization of semi-simple splicing languages in terms
of certain directed graphs.
Using this, the authors proved
that all semi-simple splicing languages must have a constant
word.
By applying one of Head's results, semi-simple splicing languages are SLT languages \cite{ssss}.
The algebraic characterization of simple splicing languages is extended to semi-simple
splicing languages in \cite{ceterchi}.
Both in the initial paper about simple systems \cite{mat},
and later by Head (who gave the name of $k-$fat [semi-]simple H systems) splicing rules of the form
$x\#1\$x\#1$ were considered, with $|x| \leq k$ , for a given constant $k$.
In \cite{ceterchi2004}, $k-$fat semi-simple splicing systems were investigated
both for linear (and circular) strings.
These systems are a particular case of splicing systems with {\em one-sided context},
i.e., each rule has the form $u\#1\$v\#1$ or $1\#u\$1\#v$
and $R$ is reflexive.
Head stated again a relationship between splicing and constants
and proved that it is decidable whether a regular language is generated by
one-sided context splicing systems, but only when the rules are either $u\#1\$v\#1$ or $1\#u\$1\#v$ \cite{hone}.

\section{Computational Power and Decidability Questions for Linear Splicing} \label{mainlinear}

As already said, the
class of languages generated by finite splicing systems is included
in the class of regular languages. This result was firstly proved by
Culik II and Harju \cite{ch,hp06,pixDAM}.
Gatterdam  gave $(aa)^*$ as an example of a regular language which cannot
be generated by a splicing system. Thus, the class of languages generated by splicing systems is
strictly included in the class of regular languages \cite{gat}.
However, for any regular language $L$ over an
alphabet $A$, by adding a marker $b \not \in A$ to the left side of every word in $L$ we obtain the language $bL$
which can be generated by a splicing system. For instance,
the language $b(aa)^*$ is generated by
$I=\{b, baa\}$ and the rule $baa \# 1 \$ b \# a$ \cite{hone,vz}.
This led to the question of whether or not one of
the known subclasses of the regular languages
corresponds to the class of languages which can be generated by a splicing system.
In turn, we are faced with the problem of finding a characterization
of the latter class (see \cite{linDAM} for a construction of a subclass
of splicing languages). All
investigations to date indicate
that the splicing languages form a class that does not coincide with another naturally defined
language class.

A characterization of languages generated by {\em reflexive splicing systems} using
constants has
been given in \cite{b,rifl}.
A splicing system is reflexive if $R$ is reflexive.
Recently, it was proven that every splicing language
has a constant \cite{bj}.
However, not all languages which have a constant are generated by splicing
systems. For instance, in the language $L = (aa)^* + b^*$
every word $b^i$ is a constant, but $L$ is not generated
by a splicing system.

Another approach was to find an algorithm which decides whether a given regular language
is generated by a splicing system.
This problem has been investigated and partially solved in \cite{gp,hp06}: it is decidable whether a regular language
is generated by a reflexive splicing system. It is worth mentioning that a splicing system by the
original definition in \cite{h87} is always reflexive.
A related problem has been investigated:
given a regular language $L$ and a finite set of enzymes, represented by a set of reflexive rules $R$,
it is decidable whether or not $L$ can be generated from a finite set of axioms by using
only rules from $R$ \cite{kim}.
In \cite{kk} the authors settle the decidability problem, by proving that for a given regular language,
it is indeed decidable whether the language is generated by a splicing system (which is not necessarily
reflexive). The proof is constructive, i.e.,
for every regular language $L$ they prove that there exists a splicing system
$(I_L,R_L)$ and if $L$ is a splicing language,
then $L$ is generated by the splicing system $(I_L,R_L)$. The
size of this splicing system depends on the size $m$ of the syntactic monoid of $L$.
All axioms in $I_L$ and the four components of every rule in $R_L$ have
length in $O(m^2)$.
By results from \cite{hp06,hpg}, one can construct a finite automaton which
accepts the language generated by $(I_L,R_L)$. Then, by
comparing it with a finite automaton which accepts $L$,
we can decide whether $L$ is generated by a splicing system.

\section{Circular Words and Languages} \label{cwdef}

Given $w \in A^*$, a circular word
$\1 w$ is the equivalence class of $w$ with respect to
the {\em conjugacy} relation $\sim$ defined by $xy \sim yx$,
for $x,y \in A^*$ \cite{lot}.
The notations $|\1 w|$, $|\1 w|_a$,
$alph(\1 w)$ will be defined
as $|w|$, $|w|_a$, $alph(w)$,
for any representative $w$ of $\1 w$.
Analogously, we define the {\it reversal} $\1 w^{Rev}$ of the
circular word $\1 w$ by $\1 w^{Rev} = \1 (w^{Rev})$.
Notice that $\1 w^{Rev}$ does not depend on which
representative in $\1 w$ we choose to define it by.
When the context does not make it ambiguous,
we will use the notation $w$ for a circular word
$\1 w$.
For a word $w$, we set
$Fact_c(w) = \{ x \in A^+ \mid \exists w' \sim w : x \in Fact(w')\}$.
Let $\1 A^*$ denote the set of all circular words
over $A$, i.e., the quotient of $A^*$ with respect to $\sim$.
Given $L \subseteq A^*$, $\1 L=\{ \1 w ~|~ w \in L\}$
is the {\it circularization} of $L$ whereas, given
a {\it circular language} $C \subseteq \1 A^*$,
every $L \subseteq A^*$ such that
$\1 L=C$ is a {\it linearization} of $C$.
In particular, a linearization of $\1 w$ is a linearization
of $\{\1 w \}$, whereas the {\em full linearization}
$Lin(C)$ of $C$ is defined by $Lin(C)=
\{w \in A^* ~|~ \1 w \in C \}$.
Notice that, given $L \subseteq A^*$,
the notation $\1 L^*$ is unambiguous
(and means $\1 (L^*)$). The same holds for $\1 L^+$.
Furthermore, we will often write $\1 w$
instead of $\{\1 w \}$ and $L$
instead of $\1 L$, for a set of letters
$L \subseteq A$. Given a family of languages
$FA$ in the Chomsky hierarchy, $FA^\sim$ is the set of all
those circular languages $C$ which have some linearization
in $FA$. Thus $Reg^\sim$ is the class of circular languages $C$
such that  $C = \1 L$ for some $L \in Reg$.
If $C \in Reg^\sim$ then $C$ is a {\it regular circular language}.
Analogously, we can define {\it context-free}
(resp. {\it context-sensitive})
circular languages.
It is classically known that given a regular
(resp. context-free, context-sensitive) language $L \subseteq A^*$,
$Lin(\1 L)$ is regular (resp. context-free, context-sensitive) \cite{hum,kud}.
As a result, a circular language $C$ is
regular (resp. context-free, context-sensitive)
if and only if $Lin(C)$ is a
regular (resp. context-free, context-sensitive)
language \cite{hb}.

\section{Circular Splicing} \label{cdef}

\subsection{P\u{a}un Circular Splicing Systems}

As in the linear case, there are different definitions
of the circular splicing operation.
In this paper we deal with the definition of
this operation given in \cite{hb}.
The corresponding circular splicing systems are named
here {\em P\u{a}un circular splicing systems} since
they are the counterpart of P\u{a}un linear
splicing systems in the circular context.

\smallskip
{\bf \tt P\u{a}un's definition \cite{hb}.}
A {\em P\u{a}un circular splicing system} is a triple $S = (A,I,R)$,
where $A$ is a finite alphabet, $I$ is the {\em initial} circular
language, with $I \subseteq \1 A^*$ and $R$ is the set
of the {\em rules}, with $R \subseteq A^* \# A^* \$ A^* \# A^*$
and $\#, \$ \not \in A$. Then, given a rule
$r=u_1 \#u_2 \$ u_3 \# u_4$ and circular words
$\1 w'$, $\1 w''$, $\1 w$, we set $(\1 w', \1 w''){\vdash}_r \1 w$
if there are linearizations $w'$ of $\1 w'$,
$w''$ of $\1 w''$, $w$ of $\1 w$ such that $w'= u_2xu_1$,
$w''= u_4yu_3$ and $w = u_2xu_1u_4yu_3$.
If $(\1 w', \1 w''){\vdash}_r \1 w$ we say that $\1 w$ is
generated (or spliced) starting with
$\1 w'$, $\1 w''$ and by using a rule $r$. We also say
that $u_1u_2$, $u_3u_4$ are {\em sites} of
splicing and we will use $SITES(R)$ to denote the
set of sites of the rules in $R$.
\medskip

From now on, ``splicing system'' will be synonymous with
``circular P\u{a}un splicing system''. Furthermore,
a {\it finite} splicing system $S = (A,I,R)$
is a circular splicing system with both $I$ and
$R$ finite sets. We will now give the definition of
circular splicing languages.
Given a splicing system  $S$ and a circular
language $C \subseteq \1 A^*$, we set
$\sigma'(C)=\{w \in
\1 A^* ~|~ \exists w', w'' \in C, \exists r \in R : \;
(w',w''){\vdash}_r ~w \}$. Then, we define
$\sigma^0(C)=C$,
$\sigma^{i+1}(C) = \sigma^i(C) \cup \sigma'(\sigma^i(C))$, $i \geq 0,$
and $\sigma^*(C) = \bigcup_{i \geq 0} \sigma^i(C)$.

\begin{definition} [Circular splicing language]
Given a splicing system $S$, with
initial language $I \subseteq \1 A^*$,
the circular language
$L(S)=\sigma^*(I)$ is the {\em language
generated} by $S$.
A circular language
$C$ is {\em P\u{a}un generated}
(or $C$ is a {\em (circular) splicing language})
if a splicing system $S$
exists such that
$C=L(S)$.
\end{definition}

\begin{example} \cite{rairo}
The regular language
$L =\{w \in A^* ~|~ \exists h,k \in \N
~ |w|_a=2k, |w|_b=2h \}$
is the full linearization of the splicing language
generated by $S = (A,I, R)$, where
$A = \{a,b \}$,
$I = \1 \{1, aa, bb, abab \}$ and
$R = \{ 1 \# 1\$ 1 \# aa, 1 \# 1 \$ 1 \# bb,$ $1 \# 1\$ 1\#abab, 1 \#1 \$ 1\# baba \}$.
\end{example}

As observed in \cite{hb}, we may assume that
the set $R$ of the rules in a splicing system
$S = (A,I,R)$ satisfies additional conditions,
having also a biological counterpart.
Namely, we may assume that $R$ is
reflexive or $R$ is symmetric (see Section \ref{lsplicing}
for these definitions).
We do not assume that $R$ is reflexive.
On the contrary, we notice that, in view of
the definition of circular splicing, if
$(w',w''){\vdash}_r ~ w$, with
$r = u_1 \# u_2 \$ u_3 \# u_4$, then
$(w'',w'){\vdash}_{r'} ~ w$, with
$r' = u_3 \# u_4 \$ u_1 \#u_2$. Consequently,
$L(S) = L(S')$, where $S' = (A,I,R')$ and
$R' = R \cup \{u_3 \# u_4 \$ u_1 \#u_2 ~|~
u_1 \# u_2 \$ u_3 \# u_4 \in R \}$. Hence,
in order to find a characterization of the
circular splicing languages, there is no loss of
generality in assuming that $R$ is symmetric.
Thus, in what follows, we assume that
$R$ is symmetric (and we do not consider this
assumption as an additional condition). However,
for simplicity, in the examples of P\u{a}un systems,
only one of either $u_1 \# u_2 \$ u_3 \# u_4$
or $u_3 \# u_4 \$ u_1 \#u_2$ will be reported
in the set of rules.

We recall that the original definition of
circular splicing was proposed by
Head \cite{h92}. He defined circular splicing
as an operation on two circular words
$\1 ypxq$, $\1 zuxv \in \1 A^*$
performed by two triples $(p,x,q)$,$(u,x,v)$
and producing $\1 ypxvzuxq$.
The word $x$ is called a {\em crossing} of the triple.
A {\em Head circular splicing system}
$S = (A,I,T,P)$ is defined by giving
a finite alphabet $A$, the initial set
$I \subseteq \1 A^*$, the set $T$ of triples,
$T \subseteq A^* \times A^* \times A^*$,
and where $P$ is a binary relation on $T$ such
that, for each $(p,x,q), (u,y,v) \in T$,
$(p,x,q)P(u,y,v)$ if and only if $x=y$.
Another definition of circular splicing
has been given by Pixton \cite{pixDAM}. In his scheme
circular splicing is performed on two circular
words $\1 w' = \1 \alpha \epsilon$,
$\1 w'' = \1 \alpha' \epsilon'$,
by using $r = (\alpha, \alpha' ; \beta)$,
$\overline {r} = (\alpha', \alpha;\beta')$
and producing
$\1 w = \1 \epsilon \beta \epsilon' \beta'$.
A {\em Pixton circular splicing system}
$S = (A, I, R)$ is defined by giving
a finite alphabet $A$, an initial set
$I \subseteq \1 A^*$ and a set $R$ of rules,
$R \subseteq A^* \times A^* \times A^*$.
Obviously, the counterpart of the notion of a
reflexive (resp. symmetric) set of rules
can be defined for (and added to) Head and Pixton systems
as the notions of the corresponding generated circular
languages.

We also recall that in the original definition of circular
splicing given in \cite{hb}, rules in $R$
could be used in two different ways \cite{hb}.
One way has been described above
while we recall the other, known as
{\it self-splicing}, below.
Self-splicing has also a biological counterpart \cite{hb}.
It introduces a different semantics
of how rules are used. While in the case of the
circular splicing operation, two words
are pasted together to form a new circular word,
in the case of the self-splicing operation, a single
circular word gives rise to two circular words.
The precise definition is given below.

\smallskip
{\bf \tt Self-splicing.}
Let $S = (A, I, R)$ be a splicing system.
Then, given a rule $u_1 \# u_2 \$ u_3 \#u_4$
and a circular word $\1 w$, we set
$\1 w ~ {\vdash}_r ~ (\1 w', \1 w'')$
if there are linearizations $w$ of $\1 w$,
$w'$ of $\1 w'$, $w''$ of $\1 w''$
such that $w = xu_1u_2yu_3u_4$,
$w' = u_4xu_1$ and $w'' = u_2yu_3$. If
$\1 w ~ {\vdash}_r ~ (\1 w', \1 w'')$
we say that $(\1 w', \1 w'')$ is generated starting
with $\1 w$ and by using self-splicing with a rule $r$.

\subsection{The Computation Power of Circular Splicing Systems}

The computational power of circular splicing systems
depends on (a) whether $R$ is reflexive,
(b) self-splicing is taken into account,
(c) which of the three definitions (Head's,
P\u{a}un's or Pixton's definition) is considered,
(d) the level in the (circular) Chomsky hierarchy
the initial set $I$ and the set $R$ of the rules belong to.

The problem of comparing the computational power of
the three definitions of circular splicing was tackled
in \cite{rairo}, where the authors proved that
computational power increases when we substitute Head
systems with P\u{a}un systems. Pixton systems seem to have
a computational power greater than P\u{a}un systems but
this is still an open question.

It is known that if $S = (A,I,R)$
is a P\u{a}un or a Pixton circular splicing system such that
$I \in FA^\sim$, where $FA$ is a full abstract
family of languages which is closed under cyclic closure,
$R$ is a finite, reflexive and symmetric
set of rules and self-splicing is used,
then $L(S) \in FA^\sim$ \cite{hb,pixAFL}.
In particular this result applies when $I$ is
a regular (resp. context-free,
recursively enumerable) circular language.
However, the problem of characterizing the corresponding
generated circular languages remains open in all these cases.

Unless differently stated, in this paper we deal with finite
P\u{a}un systems without the reflexivity assumption
and where the self-splicing operation is not
allowed, and with the corresponding class
of generated circular languages. It is known that this class
is incomparable with the class of regular circular languages.
Indeed, $\1 (a^2)^*a$ and $\1 ((A^2)^* \cup (A^3)^*)$ are examples
of regular circular languages which are not splicing languages
(for the latter language, this remains true even if we choose
Pixton systems) \cite{rairo,damCir}.
On the other hand, a non-context-free splicing language has been
exhibited in \cite{flat}. Moreover, in the same paper the authors
proved that splicing languages are all
context-sensitive and that it remains true even if $I$ is assumed
to be context-sensitive. These results will be thoroughly discussed in
Section \ref{Sflat}.

\subsection{Decidability Questions}\label{decid-question}

As for linear systems, the following decidability questions may be
asked. In the circular case they are
all still open.

\begin{problem}[P1] \label{P1}
Given a splicing system $S$, can we decide whether the corresponding
generated language $L(S)$ is
regular?
\end{problem}

\begin{problem}[P2] \label{P2}
Given a regular language $L$, closed under the conjugacy
relation, can we decide whether $L$ is the full linearization
of a splicing language?
\end{problem}

A related problem has been solved in \cite{flat}
(see Section \ref{Sflat}, Theorem \ref{theodecidable}).

\begin{problem}[P3] \label{P3}
Can we characterize the structure of the regular circular languages
which are splicing languages?
\end{problem}

The above problems have been solved for unary languages (see
Section \ref{one}).
Moreover they may be tackled
for special classes of splicing systems, namely alphabetic, marked and complete
systems (see Sections \ref{Sflat} -- \ref{Scomplete}).
The known results are summarized
in the following table. For each of the above problems $P1, P2, P3$,
the array below
indicates whether the answer is positive for the corresponding class of
splicing systems.

\bigskip
\bigskip

\begin{tabular}{|l|c|c|c|c|}\hline
 & $\Card(A) = 1$ &		
 alphabetic & marked & complete \\ \hline		
	P1		&  yes & ? & yes & yes 	\\ \hline		
	P2		&	yes	& yes & yes & ? \\ \hline
P3		&	yes	& ? & yes & ? \\ \hline
\end{tabular}

\subsection{The Case of a One-letter Alphabet} \label{one}

Unary languages are the simplest case that we can investigate when considering
Problems \ref{P1}--\ref{P3}.
As recalled below, the class of the P\u{a}un generated
languages on a one-letter alphabet
is a proper subset of the class of regular (circular) languages
\cite{rairo}.
In the following proposition, $\Z / n \Z$ denotes the cyclic group of order $n$
and, for $G \subseteq \N$,
we set $a^G = \{ a^g ~|~ g \in G \}$.

\begin{proposition} \label{CaPa1}
A subset $L= \1 L$ of $a^*$ is
P\u{a}un generated if and only
if either $L$
is a finite set
or there exist a finite subset $L_1$ of $a^*$,
positive integers $p, r, n$,
with $n = pr \geq 2$ and a subgroup
$G' = \{ pk ~|~ k \in \N, ~
0  \leq k \leq r - 1 \}$
of $\Z / n \Z$
such that
$L = L_1 \cup
(a^G)^+$, where
$G = G' \! \! \! \pmod{n}$,
and $\max \{\ell ~|~ a^{\ell} \in L_1 \} <
n = \min \{m ~|~ m \in G\}$.
\end{proposition}

For example, for the regular language $L = \{a^3,a^4\} \cup \{a^6,
a^{14}, a^{16}\}^+$, we have $L = L(S)$,
where $S=(\{a \},I,R)$,
$I=\{a^3,a^4, a^6, a^{14}, a^{16}\}$ and
$R = \{a^6 \# 1 \$1 \# a^6\}$. Here $L_1 = \{a^3,a^4\}$,
$n = 6$ and $G' = \{0, 2, 4 \}$.
Of course, the above result provides a solution to Problems
\ref{P1} and \ref{P3}. A positive answer to Problem \ref{P2}
has been given in \cite{damCir}.
This result is obtained by a characterization of
the minimal finite state automaton recognizing regular
splicing languages on
a one-letter alphabet.
We end this section with a
result concerning
the descriptional complexity
of a circular splicing system
which generates
a circular language $L \subseteq a^*$.

\begin{proposition} \label{DC}
Let $L \subseteq a^*$ be a P\u{a}un generated language.
Then, there exists a (minimal) splicing system
$(\{a\},I,R)$ generating $L$ with either
$R= \emptyset$ or $R= \{ r \}$
containing only one rule.
\end{proposition}

\section{Flat Splicing}\label{Sflat}

\subsection{Definitions and First Examples}

Flat splicing systems are of interest for
proving language-theoretic results because they allow us to separate
operations on formal languages and grammars from the operation of
circular closure (circularization). It appears that proofs for linear
words are sometimes simpler because they rely directly on standard
background on formal languages.

Note that in this section, the initial languages of the splicing systems
considered are not always finite, however the sets of splicing rules
remain always finite. Note also that
most of the results in this section have their counterpart in
circular splicing.

A \emph{flat splicing system} is a
triplet ${\cal S}=(A,I, R)$, where $A$ is an alphabet, $I$ is a set of
words over $A$, called the {\em initial set}, and $R$ is a finite set
of {\em splicing rules}, which are quadruplets
$\fagnotrul{\alpha}{\beta}{\gamma}{\delta}$ of words over $A$.
The words $\alpha, \beta, \gamma$ and $\delta$ are
called the {\em handles} of the rule.

Let $r = \fagnotrul{\alpha}{\beta}{\gamma}{\delta}$
(or $\alpha \# \beta \$ \delta \#\gamma$) be a splicing
rule. Given two words $u = x \alpha\cdot\beta y$ and $v = \gamma z
\delta$, applying $r$ to the pair $(u,v)$ yields the word $w= x
\alpha\cdot\gamma z\delta\cdot \beta y $. (The dots are used only to
mark the places of cutting and pasting, they are not parts of the
words.)  This operation is denoted by $\fagnotproduc{r}{u}{v}{w}$ and is
called a {\em production}. Note that the first word (here $u$) is
always the one in which the second word (here $v$) is inserted.

\begin{example} \label{exemple_simples}
1.  Consider the splicing rule $r=\fagnotrul{ab}{c}{aa}{b}$. We have the
  production $\fagnotproduc{r}{bab \cdot cc}{aaccb}{bab \cdot aaccb \cdot cc
  }$.

2. Consider the splicing rule $\fagnotrul{b}{b}{a}{a}$. Note that we  cannot
produce the word
$b\cdot a \cdot b$ from the word $b \cdot b$ and the singleton $a$,
because the rule requires that the inserted word has at least two
letters. On the contrary, the rule
$\fagnotrul{b}{b}{1}{a}$ does produce the word $bab$ from the
words $bb$ and $a$.

3. For the rule $r=\fagnotrul{1}{b}{a}{a}$, the production
 $\fagnotproduc{r}{\cdot bbc}{aba}{aba\cdot bbc}$,
is in fact a concatenation.

4. As a final example, the rule   $\fagnotrul{1}{1}{1}{1}$
permits all insertions (including concatenations) of a word into another one.
\end{example}

\begin{remark} \label{concat_insert}
As we can see, the ${1}$ in the rules permit the insertion of a one-lettered word
(when there is at least one ${1}$ in the middle of the rule, as it is the case in
the second item  of the previous example) or the  concatenation of two words
(when there is at least one ${1}$ at  the beginning or the end of  the rule, as it is the case in
the third item  of the previous example).
\end{remark}

As for  other forms of splicing,
the \emph{language generated} by the flat splicing system ${\cal S}
=(A,I, R)$, denoted $\fagnotlang{{\cal S}}$, is the smallest language $L$
containing $I$ and closed by $R$.

\begin{example} \label{sir_ex} Consider the splicing system over $A =
  \{a, b\}$ with initial set $I= \{ab\}$ and the unique splicing rule
  $r =\fagnotrul{a}{b}{a}{b}$.  It generates the context-free and
  non-regular language $\fagnotlang{{\cal S}}= \{ a^nb^n \mid n \geq 1\}$.

\end{example}

A splicing system is \emph{finite} (resp. \emph{regular,
  context-free, context-sensitive}) if its initial set is finite
(resp. regular, context-free, context-sensitive).

A  rule $r = \fagnotrul{\alpha}{\beta}{\gamma}{\delta}$ is
\emph{alphabetic} if its four handles $\alpha, \beta, \gamma$ and
$\delta$ are letters or the empty word. A splicing system is
alphabetic if all its rules are alphabetic.


\subsection{Two General Results}


In Section \ref{decid-question}, some decidability questions are asked. In
\cite{flat}, the following result on a similar question is proved.

\begin{theorem}\label{theodecidable}
  Given a regular flat (resp. circular) splicing system ${\cal S}$ and
  a regular language $K$, it is decidable whether $\fagnotlang{{\cal S}} =
  K$.
\end{theorem}

\begin{remark}
 Note also  that
  it is decidable  whether a regular language can
  be generated by an alphabetic (flat or circular) finite splicing system.
  This fact answers a special case of Problem \ref{P2}.
\end{remark}

The highest level in Chomsky hierarchy which can be
obtained by splicing systems with a finite
initial set  is the context-sensitive level.
This result (proved in \cite{flat}) remains true when  the initial set is context-sensitive.

\begin{theorem}\label{theo-contextsensitive}
  The language generated by a context-sensitive  flat (resp. circular)
  splicing system is context-sensitive.
\end{theorem}

We give here an example of a splicing system having a finite initial set
which produces a non-context-free language.

\begin{example}
Let  $A$ be the alphabet $ \{a,b,c,d, \fagnotbb , \fagnotff  \}$ and set $u=abcd$.
Let ${\cal S}= (A,I,R)$ be the flat splicing system
with
\begin{displaymath}
I = \{ \fagnotbb u\fagnotff , a, b, c, d \}
\end{displaymath}
and with $R$ composed of the rules
\begin{align*}
  &a_1 = \fagnotrul{\fagnotbb }{u}{a}{1},
               &\hspace{-1cm}&a_2 = \fagnotrul{au}{u}{a}{1},  \\
  &b_1 = \fagnotrul{a}{u\fagnotff }{b}{1},
         &\hspace{-1cm}&b_2 = \fagnotrul{a}{uabu}{b}{1}, \\
  &c_1 = \fagnotrul{\fagnotbb ab }{u}{c}{1},
          &\hspace{-1cm}&c_2 = \fagnotrul{abcuab}{u }{c}{1} ,\\
  &d_1 = \fagnotrul{abc }{u\fagnotff }{d}{1},
            &\hspace{-1cm}&d_2 = \fagnotrul{abc}{uuu}{d}{1}\,.
 \end{align*}

 This splicing system produces the language
\begin{align*}
  \fagnotlang{{\cal S}}  = &I \cup \{ \fagnotbb  (u)^{2^n}\fagnotff   \mid n \geq 0\} \nonumber\\
  & \cup \{ \fagnotbb  (au)^p(u)^q\fagnotff   \mid p+q = 2^n, n \geq 0\} \nonumber\\
  & \cup \{ \fagnotbb  (au)^p(abu)^q\fagnotff   \mid p+q = 2^n, n \geq 0\} \nonumber \\
  & \cup \{ \fagnotbb  (abcu)^p(abu)^q\fagnotff   \mid p+q = 2^n, n \geq 0\} \nonumber\\
  &\cup\{ \fagnotbb (abcu)^p(uu)^q\fagnotff \mid p+q = 2^n, n \geq 0\}\,.
 \end{align*}

Here is an idea of how this splicing system will produce the word
$\fagnotbb u^{2^{(n+1)}}\fagnotff$ from the word $\fagnotbb u^{2^n}\fagnotff$
by  adding a word $u$ before each of its occurrences. This is done by
inserting first letters $a$ from left to right, then  letters $b$
 from right to left, then letters $c$ from  left to right, and finally
  letters $d$ from right to left.

The intersection of the language $\fagnotlang{{\cal S}}$ with the regular
language $\fagnotbb (u)^*\fagnotff $ is equal to $\{ \fagnotbb (u)^{2^n}\fagnotff \mid n \geq
0\}$. The latter language is not context-free.

Concerning circular splicing systems, if we take the circular splicing system with the
same rules and the same initial language, we get a similar result,
which is not context-free either.

More detailed explanations, for both flat and circular splicing systems,
are provided in \cite{flat}.
\end{example}


\subsection{Alphabetic Splicing Systems}\label{sec:alphabetic}


We recall that a  splicing system is said {\em alphabetic}
if all the handles of all its rules
have length at most one.

The splicing system of Example~\ref{sir_ex} is alphabetic and has a finite initial set.
However, it generates a non-regular language. We give now a similar example
for the flat and circular cases.

\begin{example} \label{dyck_ex} Let ${\cal S}=(A,I,R)$ be the flat
  splicing system defined by $A = \{a,\bar{a} \}$, $I= \{a\bar{a}\}$
  and $R = \{ \fagnotrul{1}{1}{1}
  {1} \}$. It generates the Dyck language.  Recall that the
  Dyck language over $\{a, \bar{a}\}$ is the language of parenthesized
  expressions, $a, \bar{a}$ being viewed as a pair of matching
  parentheses (see, for example, \cite{MR549481}).

  The circular splicing system ${\cal S}=(A,I,R)$ defined by $A =
  \{a,\bar{a} \}$, $I= \{\fagnotcir{(a\bar{a})}\}$ and $R = \{ 1\#1\$1\#1 \}$
  generates the language $\hat{D}$ of words having as many $a$ as
  $\bar{a}$.
 Indeed, the system allows to insert $a\bar{a}$ and $\bar{a}a$ anywhere,
 hence we get the set of words having as many $a$ as $\bar{a}$.  This
 language $\hat{D}$ is the circularization of the
  Dyck language \cite{MR549481}.
\end{example}

\begin{remark}
  All examples given so far show that alphabetic splicing systems
  generate always a context-free languages, and this is indeed the
  case as we shall see below. Observe however that we cannot get all
  context-free languages and even all regular languages
 as splicing languages with a finite initial
  set. For example, we can easily see that
  the language $L =\{ a^*c \mid n \geq 0\}$
  cannot be obtained by such a splicing system: consider indeed the
  fact that all words of $L$ have exactly only one  $c$,
 so inserting a word of $L$
into another  word of $L$ will produce a word out of $L$.
\end{remark}

The main result of this section is the following.

\begin{theorem} \label{context-fre} 
 The language generated by a flat or circular alphabetic
  context-free splicing system is context-free.
 \end{theorem}
This theorem is effective, that is, we can actually construct
context-free grammars which generate the language produced by the
splicing system. The whole proof is in \cite{flat}, but  a hint of the proof is given below, and an example is then given  to illustrate it.

\begin{sketch}

The flat case is simpler, for this reason, we begin with it.

Let ${\cal S}$ be a flat splicing system $(A, I, R)$ which generates the language $L = \fagnotlang{{\cal S}}$.

We have seen (Example~\ref{exemple_simples}) that some productions are
in fact concatenations in contrast with ``proper insertions''
 where we really put a word between two non-empty parts of another one.
 As these concatenations cannot be treated like ``proper insertions'',
we are going to separate both types of productions.
And we prove that, roughly, we can do all
concatenations before any insertions.

Then, we first construct a grammar $G_1$ that generates the language $K$  obtained
from the initial language
$I$ by performing all the concatenations iteratively.
Then, second, we construct a grammar $G_2$ that generates the language $L$ obtained
from the language $K$
by performing all the insertions iteratively.

In the circular case, we generate the full linearization of the language
$L = \fagnotlang{{\cal S}}$.
The only difference is that the calculation of the set of concatenations and the set of insertions we need
to consider, is less straightforward.
\end{sketch}
\begin{remark}
The proof makes use of the notion of generalized  context-free grammars, that is context-free grammar
which may have rules whose right parts are context-free languages upon the set of terminal
and non terminal symbols. For example, the rule  $A \to B^*c^* \mid L$, with $a$
and $B$ non-terminal variables, $c$ a terminal variable, and $L = \{a^nb^n\mid n \geq 1\}$, is
a correct rule for a  generalized  context-free grammar.
These grammars are known to produce context-free languages (see \cite{Kral70}).
 \end{remark}

\begin{example}
 Consider the flat splicing system  ${\cal S} = (A, I, R)$ over the
  alphabet $A = \{ a, b\}$, with $I = \{aa, ab\}$, and with $R$ composed of
  the splicing rules $\{  \fagnotrul{1}{a}{a}{a} ,  \fagnotrul{a}{b}{a}{b} \}$.

The grammar which generates  the language $K$  obtained
from the initial language
$I$ by performing all the concatenations iteratively, is
\begin{align}
 S  &\to \fagnotwrd aa \ \mid \ \fagnotwrd ab \label{fagnoteq_1}\\
 \fagnotwrd aa & \to aa \ \mid \ \fagnotwrd aa \ \fagnotwrd aa \label{fagnoteq_2}\\
 \fagnotwrd ab &\to  ab  \ \mid \ \fagnotwrd aa \ \fagnotwrd ab \label{fagnoteq_3}
\end{align}

In this grammar,  the symbol $\fagnotwrd{a}{b}$ is used to
  derive words of length at least~$2$ that start with the letter $a$
  and end with the letter $b$, that is the set $K \cap aA^*b$.
(Note that the one-lettered words would need other variables.)

In line \ref{fagnoteq_2}, the first rule is used to derive the one word of
the initial set that
begins and ends with the letter $a$; the second rule is used to permit
the concatenation of
two words that begin and end with the letter $a$, producing a word
that begins and ends also with the letter $a$.
The rules of line  \ref{fagnoteq_3} are similarly constructed.

One can easily check that this grammar produces the language $(a^2)^+ + (a^2)^*ab$, and
that this language is the same as the one produced by the concatenations.

The grammar which generates  the final language $L   = \fagnotlang{{\cal S}}$ obtained
from the initial language
$K$ by performing all the pure insertions iteratively, is
\begin{align}
 S  &\to \fagnotwrd aa \ \mid \ \fagnotwrd ab \label{fagnoteq_4}\\
 \fagnotwrd aa & \to ((a \fagnotins aa) ^2)^* a \fagnotins aa a\label{fagnoteq_5}\\
 \fagnotwrd ab &\to  ((a \fagnotins aa) ^2)^*a\fagnotins ab b \label{fagnoteq_6}\\
 \fagnotins aa & \to  \fagnotins aa \ \fagnotwrd aa \ \fagnotins aa \ \mid \ 1  \label{fagnoteq_7}\\
  \fagnotins ab &\to  \fagnotins aa \ \fagnotwrd ab \ \fagnotins bb \ \mid \ 1 \label{fagnoteq_8}\\
\fagnotins bb &\to   1 \label{fagnoteq_9}
\end{align}

In this grammar,  the symbol $\fagnotwrd{a}{b}$ has the same use as in the preceding one.
The symbol $\fagnotins{a}{b}$  is always preceded by a letter $a$ or
 a non-terminal which
eventually derives a word ending with a letter $a$ and, similarly, the same
symbol $\fagnotwrd{a}{b}$ is always followed
by a letter $b$  or a non-terminal which
eventually derives a word beginning with a letter $b$.
Note that in such a grammar, you can always derive words with variables
of type $\fagnotwrd{x}{y}$ between each other type of letters.
This is intended to simulate the
insertions.

In line \ref{fagnoteq_8}, the first rule is used to permit the insertion  of a
word that begins  with the letter $a$ and ends with the letter $b$
in another word
between a letter $a$ and a letter $b$. Note that the  $\fagnotins aa$
of the same rule will possibly be used later to make new insertions,
while  $\fagnotins bb$ does not produce anything else than $1$ because
no rule permits an insertion between two letters $b$.

One can easily check that this latter grammar produces the language
$(a^2)^+ \cup \{a^{2n_1+1}a^{2n_2+1}\dots a^{2n_p+1}b^p\mid n_i \geq 0 \}$  which is equal to
 $(a^2)^+ \cup\{a^{p+2n}b^p\mid p\geq 1, n \geq 0 \}$. It can also be checked
that this language is the same as the one produced by the splicing system.
\end{example}

\section{Circular Simple and Semi-simple Splicing Systems} \label{cssh}

A special class of alphabetic splicing systems, namely
the \textit{P\u{a}un circular semi-simple splicing systems}
(or \textit{CSSH systems}) has been previously considered in \cite{ceterchi,ceterchi2004,ssd},
once again as
the circular counterpart of linear semi-simple splicing
systems introduced in \cite{ssss}.
$S = (A, I, R)$ is a CSSH system when
both $u_1u_2$,
$u_3u_4$ have length one.
Thus, either $u_1$ (resp. $u_3$) or $u_2$ (resp. $u_4$) is $1$, the empty word.

In a CSSH system, a rule is defined by a
pair of letters and by the positions of these letters in the
rule. As in the linear case, there are four types of rules, namely
$a_i \# 1 \$ a_j \# 1$, $a_i \# 1 \$ 1 \# a_j$,
$1 \# a_i \$ a_j \# 1$ and $1 \# a_i \$ 1 \# a_j$,
with $a_i, a_j \in A$.
Furthermore, since $R$ is symmetric,
if $a_i \# 1 \$ 1 \# a_j \in R$ then we also have
$1 \# a_j \$ a_i \# 1 \in R$.
The positions of the letters in the rule play an important
role that cannot be ignored.
Thus, using the terminology of \cite{ceterchi2004}, an
{\it $(i,j)$-CSSH system}, with  $(i, j) \in \{(1, 3), (2, 4) \}$,
(resp. a {\it $(2,3)$-CSSH system}) is a CSSH system
where for each $u_1 \# u_2 \$ u_3 \# u_4 \in R$
we have $u_i$, $u_j \in A$ (resp. $u_2$, $u_3 \in A$
or $u_1$, $u_4 \in A$). For instance,
$S=(A,I,R)$, with $A = \{a,b,c \}$, $I = \1 \{aac, b \}$,
$R =\{c \# 1 \$ b \# 1 \}$ is a
$(1,3)$-CSSH system, whereas $S' = (A,I,R')$,
with $R' =\{1 \# c \$ 1 \# b \}$ is a $(2,4)$-CSSH
system \cite{ceterchi2004}. Notice that in a $(1,3)$-CSSH system,
circular splicing can be rephrased as
follows: given a rule $a_i \# 1 \$ a_j \# 1$
and two circular words $\1 xa_i$, $\1 ya_j$,
the circular splicing yields as a result $\1 xa_i ya_j$.

The special case $u_1u_2 = u_3u_4 \in A$ ({\em simple
systems}) was considered in \cite{ceterchi2003},
once again as the circular counterpart of the linear
case investigated in \cite{mat}.
Given $(i, j) \in \{(1, 3), (2, 4), (2,3) \}$,
an {\it $(i,j)$-circular simple system}
is an $(i,j)$-CSSH system which is simple \cite{ceterchi2003}.
As stated in the previous section, all these systems generate context-free languages
and there are circular simple systems
generating non-regular circular languages.
Indeed, let $S = (A,I,R)$ be the $(1,3)$-circular simple system
defined by $A = \{a, b, c\}$, $I= \1 \{baca\}$ and
$R = \{a \# 1 \$ a \# 1 \}$.
Then, $Lin(L(S)) \cap (ba)^*(ca)^* =
\{(ba)^n(ca)^n ~|~ n \geq 1\}$
and consequently $L(S)$ is not a regular circular language
\cite{nat}. A characterization of simple systems having only one rule
and generating a regular circular language has been given in \cite{completi}
(see Section \ref{Scomplete}).

In \cite{ceterchi2003},
the authors compared the classes of circular languages
generated by $(i,j)$-circular simple systems, for
different values of the pair $(i,j)$.
A precise description of the relationship among
these classes of languages along with some of
their closure properties was given.
In particular, in \cite{ceterchi2003}, the authors
proved that $(1,3)$- and
$(2,4)$-circular simple systems generate the same
class of languages.
An analogous viewpoint was adopted for P\u{a}un circular
semi-simple splicing systems in \cite{ceterchi2004} where the authors
highlighted further differences between
circular simple and CSSH systems. In particular,
the class of languages generated
by $(1,3)$-CSSH systems is not comparable with the class of
languages generated by $(2,4)$-CSSH systems.
Indeed, in \cite{ceterchi2003}, the authors proved that
there is no $(2,4)$-CSSH system $S_1$ such that
$L(S) = L(S_1)$, where $S = (A,I,R)$ is
the $(1,3)$-CSSH system defined by $A = \{a, b, c \}$, $I =
\1 \{aac, b \}$ and $R = \{c \# 1 \$ b \# 1 \}$
(see also \cite{nat} for an alternative proof of this statement).
In \cite{marked}, the authors
show that the map $\mu$ defined by $\mu(C) = C^{Rev}$
is a bijection between the class of circular
languages generated by $(2,4)$-CSSH and the
class of circular languages generated by $(1,3)$-CSSH
systems.
However, a still open problem is to find
a characterization of the class of regular circular
languages generated by CSSH systems.
The structure of these languages is unknown
even if we restrict ourselves to languages
generated by circular simple splicing systems.
This problem has been solved for special classes of CSSH systems
as we will see in the next part of this paper.
We will also make some assumptions on a CSSH system
$S = (A, I, R)$.
Firstly, it has been proved that adding the empty word
to the initial set $I$ will only add the empty word to
the language generated by $S$ \cite{rairo}.
Thus, we assume $1 \not \in I$.
Secondly we assume that any rule $r = u_1 \# u_2 \$ u_3 \#u_4$ in $R$ is {\rm useful}
(i.e., there exist
$\1 x, \1 y \in I$, such that
$u_1u_2 \in Fact_c(x)$, $u_3u_4 \in Fact_c(y)$)
and $|w|_{SITES(R)} \not = 0$, for any $w \in I$.
Indeed omitting rules or circular
words in $I$ which do not intervene in the application
of the splicing operation will not change
the language generated by a CSSH system, beyond the finite set of
words removed from $I$. This result was incorrectly stated for
P\u{a}un circular splicing systems in \cite{marked} but it is not
difficult to see that it holds
for CSSH systems.

\section{Extended Marked Systems} \label{Smarked}

Problems \ref{P1}--\ref{P3} have been solved in \cite{marked}
for the class of the {\it extended marked systems}.
They are $(1,3)$-CSSH systems $S = (A, I, R)$ such that
each $w \in I$ contains at most
one occurrence of a letter in
$SITES(R)$. In the same paper, it has been proved that in
order to solve these problems,
the assumption on
$I$ can be replaced by the
condition $I = A = SITES(R)$. A $(1,3)$-CSSH system $S = (A, I, R)$
such that $I = A = SITES(R)$ is called a {\it marked system}.
For instance, $S = (A,I,R)$,
with $A = \{a,b,c \}$, $I = \1 \{ aac, b \}$,
$R = \{ c \# 1 \$ b \# 1 \}$, is an extended marked
system whereas $S' = (A', I',R)$, with
$A' = I' = \{c, b \}$, $R = \{c \# 1 \$ b \# 1 \}$,
is a marked system.
Given an extended marked system $S = (A, I, R)$,
there is a marked system $S' = (A', I',R)$ associated with $S$.
It is obtained by erasing any letter $a$ in $A$ or in $w \in I$
such that $a \not \in SITES(R)$.
Then $L(S)$ can be constructed
with ease by means of $L(S')$. Let us illustrate this construction over
an example. Consider the above-mentioned extended marked system $S = (A,I,R)$,
with $A = \{a,b,c \}$, $I = \1 \{ aac, b \}$,
$R = \{ c \# 1 \$ b \# 1 \}$ and the associated marked system
$S' = (A', I',R)$, with $A' = I' = \{c, b \}$,
$R = \{c \# 1 \$ b \# 1 \}$. Therefore, $L(S)$ is the circularization
of the language obtained by inserting the word $aa$ on the left
of each occurrence of $c$ in the words of $Lin(L(S'))$.
In turn,
we can easily obtain solutions
to the above-mentioned problems for
marked systems when we have these solutions for
{\it transitive marked systems}, whose definition is recalled below.
To shorten notation, in the next part of this section
$S=(I,R)$ will denote a marked system and
$L(I,R)$ will be the corresponding generated language.
Furthermore, $(a_i,a_j)$ will be an
abridged notation for a rule $r = a_i \# 1 \$ a_j \# 1$
in $R$.
Then, since the set $R$ of rules specifies a symmetric
binary relation, in a natural way a marked system $S=(I,R)$
is represented by an undirected graph $G=(I, R)$,
where $I$ is the vertex set and $R$ is the edge set.
In an undirected
graph, {\it self-loops} - edges from a vertex to
itself - are forbidden but here we do not
make this assumption.
From now on, $G$ will be referred to as the graph
{\it associated} with the marked system $S=(I,R)$.
A marked system $S = (I,R)$ is {\it transitive}
if $G$ is connected and either $|I| > 1$
or $G$ is a single vertex with a self-loop.
A subset $J$ of $I$ is transitive
if the subgraph $G^J$ {\it induced} by $J$ inherits
the above-mentioned property of $G$ (i.e.,
$G^J$ is connected and either $|J| > 1$
or $G^J$ is a single vertex with a self-loop).
In this case, $S^{J}=(J,R^{J})$, with $R^{J}=R\cap (J\times J)$,
denotes the subsystem of $S$ associated with $J$.
We say that $S^{J}$ is a {\it maximal} transitive subsystem
of $S$ if $J$ is a maximal object in the class of the transitive subsets of $I$
for the order of set inclusion.
The decomposition
$\{G_h = (I_h,R^{I_h}) ~|~ 1 \leq h \leq g\}$
of $G$ into its connected components
defines the canonical decomposition
$\{ (I_h, R^{I_h}) ~|~ 1 \leq h \leq g \}$
of $S$ in disjoint maximal transitive marked subsystems
$S^{I_{h}}$. Moreover the following result holds \cite{marked}.

\begin{proposition} \label{CD}
The language $L(I,R)$ generated by the marked system $S = (I,R)$ is the disjoint union of the languages $L(I_{h},R^{I_{h}})$ generated by the maximal transitive subsystems of its canonical decomposition $\{ (I_h, R^{I_h}) ~|~ 1 \leq h \leq g \}$.
\end{proposition}

\begin{example}
Let $I=\{a_{1},a_{2},a_{3}\}$ and $R=\{(a_{1},a_{1}),(a_{2},a_{2}),(a_{2},a_{3})\}$.
Then the canonical decomposition of $S = (I,R)$
is given by the transitive marked systems $S^{1}=(I_{1},R^{I_{1}})$,
with $I_{1}=\{a_{1}\}$ and $R^{I_{1}}=\{(a_{1},a_{1})\}$, and
$S^{2}=(I_{2},R^{I_{2}})$, with $I_{2}=\{a_{2},a_{3}\}$ and $R^{I_{2}}=\{(a_{2},a_{2}),(a_{2},a_{3})\}$.
\end{example}

Proposition \ref{CD} shows that,
in order to characterize the language generated by a marked system, it is sufficient to characterize
the languages generated by the transitive marked systems of its canonical decomposition.
Therefore, in the next part of this section $S = (I,R)$ will be a transitive marked system.
Some necessary definitions are given below.

\begin{definition}
Let $S=(I,R)$ be a transitive marked system.
The \emph{distance} $d_I(a_i,a_j)$
between $a_i, a_j \in I$ is defined by
$d_I(a_i,a_j) = \min \{k ~|~ \exists b_1, \ldots , b_k
\in I ~:~ (b_h,b_{h+1}) \in R, ~ 1 \leq h \leq k -1, ~ b_1=a_i,
~b_k=a_j \}$.
The \emph{global diameter} of $S$ is defined as $\delta(S)=\max\{d_{I}(a,b)\mid a,b\in I\}$. The \emph{local diameter} of $S$ is defined as $\delta_{\ell}(S)=\max \{\delta(S^{J}) \mid \mbox{$J$ is a transitive subset of $I$}\}$.
\end{definition}

\begin{example}
 Let $S=(I,R)$ with
\begin{eqnarray*}
I &=& \{a_{1},a_{2},a_{3},a_{4},a_{5}\}, \\
R &= & \{(a_{1},a_{2}),(a_{2},a_{3}),(a_{3},a_{4}),
(a_{4},a_{5}),(a_{1},a_{5}),(a_{2},a_{5}),(a_{3},a_{5}),(a_{4},a_{5})\}.
\end{eqnarray*}
 The maximum distance between two letters in $I$ is $3$, hence the global diameter of $S$
 is $\delta(S)=3$, while the local diameter of $S$ is $\delta_{\ell}(S)=4$,
 since $\delta(S^{J})=4$ for $J=\{a_{1},a_{2},a_{3},a_{4}\}$.
\end{example}

\begin{remark}\label{rem:J}
For every transitive marked system $S$ it holds  $\delta_{\ell}(S)\geq \delta(S)$. Moreover, if $\delta_{\ell}(S)\geq 4$, then there exists a transitive subset $J=\{a_{i},a_{j},a_{k},a_{l}\}$ of $I$ such that $\delta(S^{J})=4$.
\end{remark}

\begin{theorem} \label{CR}
$L(I,R)$ is a regular circular language if and only if $\delta_{\ell}(S) \leq 3$.
\end{theorem}

Theorem \ref{CR} is a direct consequence of Lemmas \ref{lem:chain} and
\ref{lem:reg}, stated below. Proposition \ref{bounded} is well known
and also needed \cite{EHRtcs}.

\begin{proposition} \label{bounded}
Let $L$ be a regular language. There exists an integer $N$ such that if $uv \in L$,
then there is $v' \in A^*$ such that $|v'| \leq N$ and $uv' \in L$.
\end{proposition}

\begin{lemma}\label{lem:chain}
 Let $S=(I,R)$ be a transitive marked system. If $\delta_{\ell}(S)> 3$, then $L(I,R)$ is not a regular language.
\end{lemma}

\begin{sketch}
By Remark \ref{rem:J}, we can suppose that there exists a transitive subset $J$ of $I$ such that $\delta(S^{J})=4$. More precisely, we can suppose that, up to renaming letters, $I$ contains the subset $J=\{a_{1},a_{2},a_{3},a_{4}\}$ and that $\{(a_{1},a_{2}),(a_{2},a_{3}),(a_{3},a_{4})\}\subseteq R$, while neither of $(a_{1},a_{3}),(a_{1},a_{4}),(a_{2},a_{4})$ is in $R$, so that $d_{J}(a_{1},a_{4})=4$.
On the contrary, assume that $L(I,R)$ is regular. Thus, $Lin(L(J,R^J)) =
Lin(L(I,R))\cap J^{*}$ is a regular language.

For every $n>0$, consider the word $\1 w_{n}=\1 (a_{1}a_{4})^{n}$, so that $|\1 w_{n}|=2n$. It can be shown, by induction on $n$, that any word $z_n = \1 w_ny_n$ in $L(J,R^J)$
has length greater than or equal to $4n$ (actually, there is
a unique word of minimal length satisfying these conditions and it is the word $\1 (a_{1}a_{4})^{n}(a_{3}a_{2})^{n}$). This is in contradiction with Proposition \ref{bounded}.
\end{sketch}

\begin{lemma} \label{lem:reg}
Let $S=(I,R)$ be a transitive marked system with $\delta_{\ell}(S) \leq 3$ and let $w \in I^+$.
If $alph(w)$ is transitive and $|w| \geq 2$,
then $\1 w \in L(I,R)$. Consequently we have:

\begin{eqnarray*} \label{struttura}
L(I,R) & = & I ~ \cup \bigcup_{J \subseteq I, \; J
\mbox{ \scriptsize{transitive}}} \1 (\cap_{a \in J}
J^*aJ^*).
\end{eqnarray*}
\end{lemma}

\begin{sketch}
Let $w \in I^+$, with $|w| \geq 2$.
It can be proved by induction on $|w|$ that $\1 w \in L(I,R)$
if and only if $J = alph(w)$ is transitive.
Clearly, $alph(w) = J$
if and only if $w \in \cap_{a \in J} J^*aJ^*$, i.e.,
$\{ w \in J^* ~|~ |w|_{a} > 0$, for all $a \in J \} =
\cap_{a \in J} J^*aJ^*$.
\end{sketch}

Recall that a \emph{cograph} is a $P_{4}$-free graph, i.e., a graph which does not contain the path $P_{4}$ on 4 vertices as an induced subgraph. Cographs have been deeply investigated in graph theory, and linear-time recognition algorithms have been provided for this
class of graphs \cite{spinrad}.
As a corollary of Theorem \ref{CR}, a marked system $S$
generates a regular circular language if and only if its associated graph
is a cograph \cite{nat}.
Finally, we may consider
\emph{marked systems with self-splicing}, that is marked systems $S=(I,R)$ in which the self-splicing operation is allowed. The language $\overline{L}(I,R)$
generated by a marked system with self-splicing $S=(I,R)$
is defined by taking into account both splicing and self-splicing \cite{hb}.
In \cite{nat}, it has been proved that
a marked system with self-splicing always generates a regular circular language, which has a very simple structure.

\section{Monotone Complete Systems} \label{Scomplete}

For a special subclass of CSSH systems, we may characterize those
of them that generate regular circular languages, namely
{\it $(i,j)$-complete systems} (or monotone systems).
They were introduced in \cite{completi}
and their definition is recalled below.

\begin{definition} \label{completeCSSH}
An {\rm $(i,j)$-complete system} $S = (A, I, R)$ is a
finite system such that:
\begin{itemize}
\item[(1)]
$S$ is an $(i,j)$-CSSH system (i.e.,
there are fixed positions $i \in \{1,2\}$
and $j \in \{3,4\}$ such that for all
$r = u_1 \# u_2 \$ u_3 \# u_4 \in R$
we have $u_1u_2=u_i \in A$ and
$u_3u_4=u_j \in A$);
\item[(2)]
$S$ is a {\rm complete system}, i.e.,
for each $a,b \in A$, there is a rule
$u_1 \# u_2 \$ u_3 \# u_4 \in R$
such that $u_i = a$, $u_j = b$.
\end{itemize}
$S$ is a monotone complete system if there are
$i \in \{1,2\}$
and $j \in \{3,4\}$ such that $S$ is an $(i,j)$-complete system.
\end{definition}

\begin{example} \label{Ex0}
{\rm Let $S = (A, I, R)$, where $A = \{a, b \}$,
$I = \1 \{ab \}$ and
$R = \{a \# 1 \$ a \# 1,$ $b \# 1 \$ b \# 1,$ $a \# 1 \$ b \# 1 \}$.
Therefore $S$ is a $(1,3)$-complete system.
Analogously, let
$R' = \{1 \# a \$ 1 \# a,$  $1 \# b \$ 1 \# b,$ $1 \# a \$ 1 \# b \}$,
$R'' = \{1 \# a \$ a \# 1, 1 \# b \$ b \# 1, 1 \# a \$ b \# 1 \}$.
Then, $S' = (A, I, R')$ is a $(2,4)$-complete system
and $S'' = (A, I, R'')$ is a $(2,3)$-complete system.}
\end{example}

In \cite{completi}, the authors proved that the
class of languages generated
by $(i,j)$-complete systems is equal to the class
of languages generated by $(i',j')$-complete systems,
for $(i,j), (i',j') \in \{(1,3), (2,4), (2,3) \}$,
$(i,j) \not = (i',j')$.
Moreover, they proved that these systems
have the same computational
power as simple systems
with only one rule of a specific type.
Finally we recall below
the characterization of monotone complete systems
generating a regular circular language
given in \cite{completi}.
In view of the above mentioned results,
some assumptions will be made. Precisely,
in what follows $S = (A, I, R)$ will be a $(1,3)$-complete system,
with $1 \not \in I$. Any rule
$a_i \# 1 \$ a_j \# 1$ in $R$ will be denoted by the pair of
letters $(a_i, a_j)$ and $R$ may be identified with $A \times A$.
Moreover $alph(I) = A$.
The above characterization has been obtained thanks to the close relation
between complete systems and a class of context-free languages introduced
in \cite{EHRtcs}, whose definition is recalled below.

\subsection{Pure Unitary Languages} \label{completiTCS}

A {\it pure unitary grammar} is a pair $G = (A, Y)$, where
$A$ is a finite nonempty alphabet and $Y \subseteq A^+$ is a finite
set. Then, we consider the set of productions
$\{1 \rightarrow y ~|~ y \in Y \}$ and the derivation relation
$\Rightarrow^*_Y$ of the semi-Thue system
$T(Y) = \{\langle 1, y \rangle ~|~ y \in Y \}$, induced by $Y$.
We recall that $\Rightarrow^*_Y$ is the transitive and reflexive closure
of the relation $\Rightarrow_Y$, defined by $uv \Rightarrow_Y uyv$,
for any $u,v \in A^*$, $y \in Y$.

A reflexive and transitive relation on a set is called a {\it quasi-order}.
Then, $\Rightarrow^*_Y$ is a quasi-order on $A^*$ and, for brevity,
we will denote it by $\leq_Y$.
For a given quasi-order on a set $X$, one can consider
the {\it upward closure} of a subset of $X$, defined below.

\begin{definition} \label{upwardclosure} \cite{EHRtcs}
For any quasi order $\leq$ on a set $X$ and any
subset $Z$ of $X$, the {\rm upward closure} of
$Z$ (with respect to $\leq$) is given by
$cl_{\leq}(Z) = \{x \in X ~|~ \exists y \in Z$
such that $y \leq x \}$. $Z$ is $\leq${\rm-closed}
(or simply {\rm closed}) if $cl_{\leq}(Z) = Z$
(i.e., if $y \in Z$ and $y \leq x$ implies that $x \in Z$).
\end{definition}

Given a pure unitary grammar $G = (A, Y)$, the
{\it language of} $G$, denoted $L(G)$ is $cl_{\leq_Y}(1)$.
A word $w \in L(G)$ if and only if $1 \Rightarrow^*_Y w$.
A language of the form $L(G)$ for some pure unitary grammar $G$, is
called a {\it pure unitary language}.

Pure unitary languages may also be defined by
the operations of {\it insertion} and iterated insertion.
They are variants of classical operations on formal languages
and we recall their definitions below.

\begin{definition} \cite{Hau}
Given $Z, Y \subseteq A^*$, the operation of {\it insertion},
denoted by $\leftarrow$, is defined by $Z \leftarrow Y =
\{z_1yz_2 ~|~ z_1z_2 \in Z \mbox{ and } y \in Y \}$. The
operation of {\it iterated insertion}, denoted
by $\leftarrow_*$, is defined inductively from the
operation of insertion by $Y^{\leftarrow_0} = \{1 \}$,
$Y^{\leftarrow_{i+1}} = Y^{\leftarrow_{i}} \leftarrow Y$
and $Y^{\leftarrow_{*}} = \cup_{i \geq 0} Y^{\leftarrow_{i}}$.
\end{definition}

The following result has been stated in \cite{Hau} with no proof.
We give a short proof of it for the sake of completeness.

\begin{proposition} \label{insertion-unitary} \cite{Hau}
A language $L \subseteq A^*$ is a pure unitary language
if and only if $L = Y^{\leftarrow_{*}}$ with $Y$ being a
finite set of nonempty words in $A^*$.
\end{proposition}
\begin{proof}
Let $Y$ be a finite set of nonempty words in $A^*$ and let $L = Y^{\leftarrow_{*}}$.
Consider the pure unitary grammar $G = (A, Y)$. It is easy to see that a word $w$ is in $L$ if
and only if $w$ is in $L(G)$. Of course this is true for the empty word. Otherwise, assume that
$w \in  Y^{\leftarrow_{i+1}}$. Therefore there are $z_1z_2 \in Y^{\leftarrow_{i}}$ and $y \in Y$
such that $w = z_1yz_2$. By using induction on $i$, we have $z_1z_2 \in L(G)$,
hence, by the definitions, $1 \Rightarrow^*_Y z \Rightarrow_Y w$
and we conclude that $w \in L(G)$.
Conversely, if $w \in L(G)$, then there is $z$,
with $z \not = w$, and such that $1 \Rightarrow^*_Y z \Rightarrow_Y w$.
Using once again induction, we have $z \in Y^{\leftarrow_{i}}$ and by the definition of
$\Rightarrow_Y$, we have $w \in  Y^{\leftarrow_{i+1}} \subseteq Y^{\leftarrow_{*}}$.

Let  $Y$ be a finite set of nonempty words in $A^*$ and let
$L = L(G)$ the language of the pure unitary grammar $G = (A,Y)$.
The same argument as below shows that $L = Y^{\leftarrow_{*}}$.
\end{proof}

The construction of a context-free grammar generating $Y^{\leftarrow_{*}}$
is obvious (see for instance \cite{completi}).
The following result has been proved in \cite{completi}.

\begin{lemma} \label{insertion}
For each $w, z \in Y^{\leftarrow_{*}}$, for
each $w_1, w_2 \in A^*$ such that
$w = w_1 w_2$, we have
$w_1 z w_2 \in Y^{\leftarrow_{*}}$.
\end{lemma}

We recall below
a characterization of regular pure unitary languages,
given by means of a decidable property, stated in
\cite{EHRtcs}.

\begin{definition} \label{unavoidable}
Let $A$ be an alphabet, let $Y \subseteq A^*$ be a finite set.
$Y$ is \textit{unavoidable} in $A^*$ if
there exists $k_0 \in \N$ such that any word $x \in A^*$,
with $|x| > k_0$, has a factor $y$ in $Y$. The integer
$k_0$ is called subword avoidance bound for $Y$.
\end{definition}

\begin{theorem} \label{RegCar}
For any pure unitary grammar $G = (A, Y)$, the language $L(G)$ is
regular if and only if $Y$ is unavoidable in $A^*$.
\end{theorem}

Theorem \ref{RegCar} gave a necessary and sufficient condition
on a pure unitary grammar that guarantee the language of the grammar to be regular.
In the proof of this theorem, Ehrenfeucht, Haussler and Rozenberg
stated two other important intermediate results, both recalled below.
It would be interesting to find a proof of the former theorem independent of the latter
results.
We recall that a quasi-order $\leq$ on $A^*$ is {\rm monotone}
if $u \leq v$ and $u' \leq v'$ imply that $u u' \leq v v'$,
for all $u, v, u', v' \in A^*$.
Moreover, a quasi-order $\leq$ on $A^*$ is a {\rm well quasi-order} if
for each infinite sequence $\{x_i \}$
of elements in $A^*$, there exist $i < j$ such that
$x_i \leq x_j$.

\begin{theorem} \label{HigmanGen}[(generalized Higman theorem)]
For any finite set $Y \subseteq A^+$, the quasi-order $\leq_Y$ is a well quasi-order on
$A^*$ if and only if $Y$ is unavoidable in $A^*$.
\end{theorem}

\begin{proposition} \label{myhillnerodegenEHR}[(generalized Myhill-Nerode theorem)]
A language $L \subseteq A^*$
is regular
if and only if it is $\leq$-closed under
some monotone well quasi-order $\leq$ on
$A^*$.
\end{proposition}

\subsection{Pure Unitary Grammars and Monotone Complete Splicing Systems}

In \cite{completi} the authors stated the following result:
a circular language $L$ is generated by a $(1,3)$-complete system $S = (A, I, R)$
if and only if there exists a finite language $Y \subseteq A^+$, closed under the
conjugacy relation and such that
$Lin(L) = Y^{\leftarrow_{*}} \setminus \{1 \}$.
We differently state the same result below.

\begin{theorem} \label{Equiv1}
The following conditions are equivalent:
\begin{itemize}
\item[(1)]
There exists a $(1,3)$-complete system $S = (A, I, R)$ such
that $L = L(S)$.
\item[(2)]
There exists a flat splicing system
${\cal S}=(A, Y, R')$, where $Y \subseteq A^+$ is a
finite language closed under the conjugacy
relation and $R' = \{ \fagnotrul{a}{1}{1}{b} ~|~ a,b \in A \}$, such that
$\fagnotlang{{\cal S}} = Lin(L)$.
\item[(3)]
There exists a finite language $Y \subseteq A^+$
such that $Y$ is closed under the
conjugacy relation and
$Lin(L) = Y^{\leftarrow_{*}} \setminus \{1 \}$.
\end{itemize}
\end{theorem}

Theorem \ref{Equiv1} is a direct consequence
of the following two results.

Recall that in a circular splicing system
$S = (A, I, R)$, the set $R$ is supposed to be symmetric.

\begin{proposition} \label{PEquiv1}
Let $S = (A, I, R)$ be a $(1,3)$-CSSH system.
Then the flat splicing system
${\cal S}=(A, Y, R')$, where $Y = Lin(I)$
and $R' = \{ \fagnotrul{a}{1}{1}{b} ~|~ (a,b) \in R \}$, is such that
$\fagnotlang{{\cal S}} = Lin(L(S))$.
Conversely, let
${\cal S}=(A, Y, R')$ be a flat splicing system,
where $Y \subseteq A^+$ is a
finite language closed under the conjugacy
relation, $R' = \{\fagnotrul{a}{1}{1}{b} ~|~ (a,b) \in R \}$
and $R$ is a symmetric relation on $A$.
Let $I = \1 Y$ be the circularization of $Y$.
Then $\fagnotlang{{\cal S}} = Lin(L(S))$,
where $S = (A, I, R)$ is a $(1,3)$-CSSH system.
\end{proposition}
\begin{proof}
Let $S = (A, I, R)$ be a $(1,3)$-CSSH system.
Consider the flat splicing system
${\cal S}=(A, Y, R')$, where $Y = Lin(I)$
and $R' = \{ \fagnotrul{a}{1}{1}{b} ~|~ (a,b) \in R \}$.
We prove that $\fagnotlang{{\cal S}} = Lin(L(S))$. Let $L = L(S)$.
We prove first the inclusion $Lin(L) \subseteq \fagnotlang{{\cal S}}$.
The proof is by induction on the minimal number of steps used for
generating $\1 w \in L$. If the number of steps is null, we have
$\1 w \in I$, thus $Lin(\1 w) \subseteq Lin(I) = Y \subseteq \fagnotlang{{\cal S}}$.

Suppose now that for any word $\1 w \in L$ generated in at most $k$ steps,
we have $Lin(\1 w) \subseteq \fagnotlang{{\cal S}}$. Let $\1 w$ be a word
generated in at least $k+1$ steps. By the definition of the circular splicing
operation, there are two circular words $\1 u$ and $\1 v$, generated in at most
$k$ steps, a rule $(a, b) \in R$ and words $x,y$ such that $\1 u = \1 xa$,
$\1 v = \1 yb$, $\1 w = \1 xayb$. By induction,
$Lin(\1 xa) \subseteq \fagnotlang{{\cal S}}$,
$Lin(\1 yb) \subseteq \fagnotlang{{\cal S}}$ and we have to show
that any word in $Lin(\1 xayb)$ is in $\fagnotlang{{\cal S}}$.
Since $xa, yb$ are in $\fagnotlang{{\cal S}}$ and
$\fagnotrul{a}{1}{1}{b}, \fagnotrul{b}{1}{1}{a}$ are in $R'$,
then clearly $xayb, ybxa$ are in $\fagnotlang{{\cal S}}$. Hence assume
$w = st$, $xayb = ts$ for some nonempty words $s,t$, $t \not = xa$.
Therefore, either $t$ is a proper prefix of $xa$ or $xa$ is a proper
prefix of $t$. In the first case, there is a word $x'$ such that
$xa = tx'a$, $s = x'ayb$, and thus $w = x'a yb t$. Since $x'at \sim xa$, we
have $x'at \in \fagnotlang{{\cal S}}$ and so also
$w = x'a yb t \in \fagnotlang{{\cal S}}$, by using the rule
$\fagnotrul{a}{1}{1}{b} \in R'$.
Otherwise, there are words $y_1, y_2$ such that $t = xa y_1,
y = y_1y_2, s = y_2b$. Clearly $y_2by_1$ is in $\fagnotlang{{\cal S}}$
which yields $w = y_2bxa y_1  \in \fagnotlang{{\cal S}}$, by using the rule
$\fagnotrul{b}{1}{1}{a} \in R'$.
In conclusion, $Lin(\1 xayb) \subseteq \fagnotlang{{\cal S}}$.

We now prove that $\fagnotlang{{\cal S}} \subseteq Lin(L(S))$.
Let $L = \fagnotlang{{\cal S}}$.
The proof is by induction on the minimal number of steps used for
generating $w \in L$. If the number of steps is null, we have
$w \in Y = Lin(I) \subseteq Lin(L(S))$.

Suppose now that for any word $w \in L$ generated in at most $k$ steps,
we have $w \in Lin(L(S))$. Let $w$ be a word
generated in at least $k+1$ steps.
By the definition of the flat splicing
operation, there are two words $u$ and $v$, generated in at most
$k$ steps, a rule $\fagnotrul{a}{1}{1}{b} \in R'$
and words $x,y,z$ such that $u = xaz$,
$v = yb$, $w = xaybz$. By induction,
$u = xaz$, and so also $zxa$, and $yb$ are in $Lin(L(S))$.
This implies $\1 zxa, \1 yb \in L(S)$.
Since $(a,b) \in R$, by the definition of the circular splicing
operation, we also have $\1 w = \1 zxayb \in L(S)$ and
$w \in Lin(L(S)$.

Conversely, let
${\cal S}=(A, Y, R')$ be a flat splicing system,
where $Y \subseteq A^+$ is a
finite language closed under the conjugacy
relation, $R' = \{\fagnotrul{a}{1}{1}{b} ~|~ (a,b) \in R \}$
and $R$ is a symmetric relation on $A$.
Consider the $(1,3)$-CSSH system
$S = (A, I, R)$, where
$I = \1 Y$ is the circularization of $Y$.
By the first part of the proof, there is
a flat splicing system ${\cal S'}$ such that
$Lin(L(S)) = \fagnotlang{{\cal S'}}$.
Clearly ${\cal S'} = {\cal S}$ and this equality ends
the proof.
\end{proof}

\begin{proposition} \label{PEquiv2}
Let $Y \subseteq A^+$ be a set of nonempty words.
Then
$Y^{\leftarrow_{*}} \setminus \{1 \} =
\fagnotlang{{\cal S}}$, where
${\cal S}=(A, Y, R')$ is a flat splicing system and
$R' = \{ \fagnotrul{a}{1}{1}{b} ~|~ a,b \in A \}$.
\end{proposition}
\begin{proof}
We prove that $\fagnotlang{{\cal S}} \subseteq Y^{\leftarrow_{*}} \setminus \{1 \}$.
Of course $\fagnotlang{{\cal S}} \subseteq A^+$.
Let $L = \fagnotlang{{\cal S}}$.
The proof is by induction on the minimal number of steps used for
generating $w \in L$. If the number of steps is null, we have
$w \in Y  \subseteq Y^{\leftarrow_{*}} \setminus \{1 \}$.

Suppose now that for any word $w \in L$ generated in at most $k$ steps,
we have $w \in Y^{\leftarrow_{*}} \setminus \{1 \}$. Let $w$ be a word
generated in at least $k+1$ steps.
By the definition of the flat splicing
operation, there are two words $u$ and $v$, generated in at most
$k$ steps, a rule $\fagnotrul{a}{1}{1}{b} \in R'$
and words $x,y,z$ such that $u = xaz$,
$v = yb$, $w = xaybz$. By induction,
$u$ and $v$ are in $Y^{\leftarrow_{*}} \setminus \{1 \}$,
hence $w$ is also in $Y^{\leftarrow_{*}} \setminus \{1 \}$,
by Lemma \ref{insertion}.

Conversely, by induction we prove that $Y^{\leftarrow_{i}}
\subseteq \fagnotlang{{\cal S}}$, $i \geq 1$.
Clearly $Y \subseteq \fagnotlang{{\cal S}}$.
Let $w$ be a word in
$Y^{\leftarrow_{i+1}}$, $i \geq 1$.
By definition there are $z_1z_2 \in Y^{\leftarrow_{i}}$ and $y \in Y$
such that $w = z_1yz_2$. By induction the nonempty
word $z_1z_2$ is in $\fagnotlang{{\cal S}}$.
If $z_1 \not = 1$, set $z_1 = z'_1 a$, $y = y'b$, with $a,b \in A$.
Thus the word
$w = z'_1 a y'b z_2$ is in $\fagnotlang{{\cal S}}$, by using the rule
$\fagnotrul{a}{1}{1}{b} \in R'$.
If $z_1 = 1$, then $z_2 \not = 1$ and by a symmetric argument
we prove that $w \in \fagnotlang{{\cal S}}$.
\end{proof}

Theorem \ref{Equiv1} yields
the following characterization of the $(1,3)$-complete systems
generating regular circular languages.

\begin{theorem}
A $(1,3)$-complete system $S = (A, I, R)$
generates a regular circular language if and only if $Lin(I)$ is
unavoidable in $A^*$.
\end{theorem}

\begin{example}
Let $S = (\{a,b\}, \1 \{aa, b \}, \{(a,a), (b, b), (a,b) \})$.
Then $Lin(I) = \{aa, b \}$ is
unavoidable in $A^*$ \cite{lotN}. It is easy
to see that
$Lin(L(S)) = \{ w \in \{a,b\}^+ ~|~ |w|_a = 2k, k \geq 0 \}$.
\end{example}


\begin{thebibliography}{9999999}
%
\bibitem{MR549481} J. Berstel, {\it Transductions and Context-free Languages},
B. G. Teubner, Stuttgart, (1979), Available at
\url{http://www-igm.univ-mlv.fr/ ~ berstel/LivreTransductions/LivreTransductions.html}
%
\bibitem{flat}
J. Berstel, L. Boasson, I. Fagnot, Splicing systems and the Chomsky hierarchy,
{\it Theoretical
Computer Science} {\bf 436} (2012) 2-22.
%
\bibitem{bpr}
J. Berstel, D. Perrin, C. Reutenauer,
{\it Codes and Automata}, Encyclopedia
on Mathematics and its Applications
{\bf 129},
Cambridge University Press, 2009.
%
\bibitem{b}
P. Bonizzoni, Constants and label-equivalence: a decision procedure for reflexive
regular splicing languages, {\it Theoretical
Computer Science}  {\bf  411} (2010) 865-877.
%
\bibitem{bj}
P. Bonizzoni, N. Jonoska, Regular splicing languages must have a constant, in:
G. Mauri, A. Leporati (Eds.) {\it Proc. DLT 2011},
Lecture Notes in
Computer Science {\bf 6795} (2011) 82-92.
%
\bibitem{rifl}
P.Bonizzoni, C. De Felice, R. Zizza,
The structure of reflexive regular splicing
languages via Sch\"{u}tzenberger constants,
{\it Theoretical
Computer Science} {\bf 334} (2005), 71-98.
%
\bibitem{completi}
P. Bonizzoni, C. De Felice, R. Zizza,
A characterization of (regular) circular languages generated by monotone complete
splicing systems, {\it Theoretical
Computer Science} {\bf 411} (2010) 4149-4161.
%
\bibitem{rairo}
P. Bonizzoni, C. De Felice, G. Mauri, R. Zizza,
Circular splicing and regularity,
{\it Theoretical Informatics and Applications} {\bf 38} (2004) 189-228.
%
\bibitem{nat}
P. Bonizzoni, C. De Felice, G. Fici, R. Zizza, On the regularity of circular
splicing languages: a survey and new developments,
{\it Natural Computing} {\bf 9} (2010) 397-420.
%
\bibitem{damCir}
P. Bonizzoni, C. De Felice, G. Mauri, R. Zizza, On the power of circular splicing,
{\it Discrete Applied Mathematics} {\bf 150} (2005) 51-66.
%
\bibitem{linDAM}
P. Bonizzoni, C. De Felice, G. Mauri, R. Zizza,
Linear splicing and syntactic monoid,
{\it Discrete Applied Mathematics} {\bf 154} (2006) 452-470.
%
\bibitem{spinrad}
A. Brandst\"{a}dt, V. B. Le, J. Spinrad,  {\it Graph Classes: a
Survey}, SIAM Monographs on Discrete Matematics and
Applications, 1999.
%
\bibitem{ceterchi}
R. Ceterchi,
An algebraic characterization of semi-simple splicing,
{\it Fundamenta Informaticae} {\bf 73} (2006) 19-25.
%
\bibitem{ceterchi2004}
R. Ceterchi, C. Martin-Vide, K. G. Subramanian,
On some classes of splicing languages,
in: N. Jonoska, G. P\u{a}un, G. Rozenberg (Eds.),
{\it Aspects of Molecular Computing: Essays in Honor of the 70th Birthday of Tom Head},
Lecture Notes in Computer Science {\bf 2950} (2004) 83-104.
%
\bibitem{ceterchi2003}
R. Ceterchi, K. G. Subramanian,
Simple circular splicing systems,
{\it Romanian Journal of Information Science and Technology} {\bf 6} (2003) 121-134.
%
\bibitem{ch}
K. Culik II and T. Harju, Splicing semigroups of dominoes and DNA,
 {\it Discrete Applied Mathematics} {\bf 31} (1991) 261-277.
 %
\bibitem{marked}
C. De Felice, G. Fici, R. Zizza,
A characterization of regular circular languages generated by marked splicing systems,
{\it Theoretical
Computer Science} {\bf 410} (2009) 4937-4960.
 %
 \bibitem{EHRtcs}
A. Ehrenfeucht, D. Haussler, G. Rozenberg,
On regularity of context-free languages, {\it Theoretical
Computer Science} {\bf 27} (1983) 311-332.
%
\bibitem{gat}
R. W. Gatterdam, Splicing systems and regularity,
{\it International Journal of Computer Mathematics},
{\bf 31} (1989) 63-67.
%
\bibitem{ssss}
E. Goode, D. Pixton, Semi-simple splicing systems, in:
C. Martin-Vide and V. Mitrana (Eds.) {\it Where Mathematics,
Computer Science,
Linguistics and Biology meet}, 343-357,
Kluwer Academic Publ., Dordrecht, 2001.
%
\bibitem{gp}
E. Goode, D. Pixton, Recognizing splicing languages: syntactic monoids and simultaneous
pumping,  {\it Discrete Applied Mathematics} {\bf 155} (2007) 989-1006.
%
\bibitem{Hau}
D. Haussler, Insertion languages,
{\it Information Sciences} {\bf 31} (1983) 77-90.
%
\bibitem{h87}
T. Head, Formal language theory and DNA: an analysis of
the generative capacity of specific recombinant behaviours,
{\it Bulletin of Mathematical Biology} {\bf 49} (1987) 737-759.
%
\bibitem{h92}
T. Head,
Splicing schemes and DNA,
in:  {\it Lindenmayer Systems:
Impacts on
Theoretical Computer Science and Developmental
Biology},
Springer-Verlag, Berlin (1992) 371-383.
%
\bibitem{h98slt}
T. Head (1998), Splicing representations of strictly locally testable languages,
 {\it Discrete Applied Mathematics} {\bf 87} (1998) 139-147.
%
\bibitem{hone}
T. Head, Splicing languages generated with one sided context, in:  G. P\u{a}un (Ed.)
{\it Computing With Bio-molecules: Theory and Experiments}, pp. 269-282, Springer
(1998).
%
\bibitem{hsugg}
T. Head, {\it Circular Suggestions For DNA Computing},
World Scientific, 1999.
%
\bibitem{h06}
T. Head, From biopolymers to formal language theory,
in: James A. Anderson (Ed.) {\it Automata Theory with Modern Applications},
Cambridge University Press, 231-234, 2006.
%
\bibitem{h12}
T. Head, Restriction enzymes in language generation and
plasmid computing, in: Evgeny Katz (Ed.) {\it Biomolecular Computing: From Logic Systems
to Smart Sensors and Actuators}, Wiley-VCH Verlag GmbH \& Co., 245-263, 2012.
%
\bibitem{hp06}
T. Head, D. Pixton, Splicing and regularity, in:
Z. Esik, C. M. Vide, V. Mitrana (Eds.) {\it Recent
Advances in Formal Languages and Applications}, 119-147, Springer,
2006.
%
\bibitem{hpg}
T. Head, D. Pixton, and E. Goode, Splicing systems: regularity and below, in: M. Hagiya and
A. Ohuchi (Eds.) {\it Proc. DNA 8}, Lecture Notes in Computer Science {\bf 2568}, 262-268,
Springer, 2002.
%
\bibitem{hplasmid}
T. Head, G. Rozenberg, R.S. Bladergroen, C.K.D. Breek,
P.H.M. Lommerse, H.P. Spaink (2000),
Computing with DNA by operating on plasmids,
{\it Biosystems} {\bf 57} (2000) 87-93.
%
\bibitem{hb}
T. Head, G.  P\u{a}un, D. Pixton, Language theory and
molecular genetics: generative mechanisms suggested by DNA
recombination, {\it in}: G. Rozenberg, A. Salomaa (Eds.),
{\it Handbook of Formal Languages}, Vol. $2$, 295-360, Springer
Verlag, 1996.
%
\bibitem{hum}
J.E. Hopcroft, R. Motwani, J.D. Ullman,
{\it Introduction to Automata Theory, Languages, and Computation},
Addison-Wesley, Reading, Mass., 2001.
%
\bibitem{kk}
L. Kari, S. Kopecki,
Deciding whether a regular language is generated by a splicing system,
in: D. Stefanovic and A.J. Turberfield (Eds.), {\it Proc.
DNA 18},
Lecture Notes in Computer Science {\bf 7433},
98-109, Springer, 2012.
%
\bibitem{Kral70}
J. Kr{\'a}l,  A modification of a substitution theorem and some
 necessary and sufficient conditions for sets to be
 context-free, {\it Mathematical Systems Theory} {\bf 4} (1970)
 129-139.
%
\bibitem{kim}
S.M. Kim,
An algorithm for identifying spliced languages, in: T. Jiang and D. T. Lee (Eds.),
{\it Proc. COCOON},  Lecture Notes in Computer Science
{\bf 1276}, 403-411, Springer, 1997.
%
\bibitem{kud}
M. Kudlek, Languages of cyclic words, in: N. Jonoska, G. P\u{a}un, G. Rozenberg (Eds.),
{\it Aspects of Molecular Computing}, Essays dedicated to Tom Head,
Lecture Notes in Computer Science
{\bf 2950}, 278-288, Springer, 2004.
%
\bibitem{lot}
M. Lothaire, {\it Combinatorics on Words}, Encyclopedia of Mathematics
and its Applications, Addison Wesley Publishing Company, 1983.
%
\bibitem{lotN}
M. Lothaire, {\it Algebraic Combinatorics on Words}, Encyclopedia
on Mathematics and its Applications
{\bf 90},
Cambridge University Press, 2002.
%
\bibitem{mat}
A. Mateescu, G.  P\u{a}un, G. Rozenberg, A. Salomaa, Simple splicing systems,
 {\it Discrete Applied Mathematics} {\bf 84} (1998) 145-163.
%
\bibitem{Paun96}
G.  P\u{a}un, On the splicing operation,
 {\it Discrete Applied Mathematics} {\bf 70} (1996) 57-79.
%
\bibitem{book}
G.  P\u{a}un, G. Rozenberg, A. Salomaa, {\it DNA Computing, New Computing
Paradigms}, Springer-Verlag, Berlin, 1998.
%
\bibitem{pixDAM}
D. Pixton, Regularity of splicing languages,
{\it Discrete Applied Mathematics} {\bf 69} (1996) 101-124.
%
\bibitem{pixAFL}
D. Pixton,
Splicing in abstract families of languages,
{\it Theoretical
Computer Science}
{\bf 234} (2000) 135-166.
%
\bibitem{sch}
M. P. Sch\"utzenberger,
Sur certaines operations de fermeture dans le langages rationnels,
{\it Symposia Mathematica} {\bf 15} (1975) 245-253.
%
\bibitem{ssd}
R. Siromoney, K.G. Subramanian, A. Dare, Circular DNA and
splicing systems, in: {\it Proc. of ICPIA}, Lecture Notes in
Computer Science {\bf 654}, 260-273. Springer, 1992.
%
\bibitem{vz}
S. Verlan, R. Zizza, 1-splicing vs. 2-splicing: separating results,
in: T. Harju and J. Karhumaki (Eds.) {\it Proc. of
WORDS 03}, TUCS General Publication
(Turku Centre for Computer Science) {\bf  27}, 2003.
%
\bibitem{z}
R. Zizza, Splicing systems,
{\em Scholarpedia} (2010), 5(7): 9397.

\end{thebibliography}
\end{document}